\documentclass[11pt]{article}

\usepackage{amssymb,amsmath}
\usepackage[thmmarks,amsmath,amsthm]{ntheorem}
\usepackage{amsfonts,psfrag,enumitem,indentfirst,float,geometry,color,enumerate}
\usepackage{graphicx,subfigure,algorithm,algpseudocode,pifont}
\usepackage[hidelinks]{hyperref}

\usepackage{epstopdf}
\usepackage{stmaryrd}
\usepackage[normalem]{ulem}

\usepackage{booktabs}
\usepackage{array}
\usepackage{placeins}
\usepackage[font=small,labelfont=bf]{caption}
\usepackage{microtype}

\allowdisplaybreaks[4]
\usepackage{lineno}

\numberwithin{equation}{section}

\newtheorem{proposition}{Proposition}[section]

\newcommand{\bx}{\boldsymbol{x}}
\newcommand{\by}{\boldsymbol{y}}
\newcommand{\bn}{\boldsymbol{n}}
\newcommand{\cp}{\boldsymbol{\pi}}
\newcommand{\U}{\mathcal{U}}
\newcommand{\R}{\mathcal{R}}
\newcommand{\mean}[2]{\langle #1\rangle_{#2}}
\newcommand{\blfootnote}[1]{%
  \begingroup
  \renewcommand{\thefootnote}{}%
  \footnotetext{#1}%
  \endgroup
}

\hypersetup{
  pdftitle={A Unified Interface Neural Network via Projection-Correction for Elliptic Interface Problems}
}

\title{A Unified Interface Neural Network for Elliptic Interface Problems
}
\author{
Qiao Zhuang \thanks{School of Science and Engineering, University of Missouri-Kansas City, Kansas City, MO 64110,  USA  (qzhuang@umkc.edu) }\and
Xu Zhang \thanks{Department of Mathematics, Oklahoma State University, Stillwater, OK 74078, USA (xzhang@okstate.edu)}}

\date{\small \today}

\begin{document}
\maketitle
\vspace{-2.5em}
\begin{abstract}
We develop a Unified Interface Neural Network (UINN) framework for elliptic
interface problems with discontinuous coefficients. A single shared
network provides a unified representation of the solution in both
subdomains, while a projection--correction construction embeds the solution
and flux jump conditions directly into the neural trial space. This
design aligns with the classical numerical principle of incorporating
interfacial physics into the approximation basis to construct interface-unfitted numerical methods. The jump
conditions are therefore satisfied exactly for any network parameters. Compared with methods based on separate phase networks
and numerical formulations with interface penalty terms, UINN provides a unified
representation with exact enforcement of interfacial jump conditions and
interface-penalty-free optimization. Those features streamline the numerical formulations and can enhance training robustness.
Numerical experiments on smooth interfaces at modestly large and
high coefficient contrasts demonstrate UINN's solution accuracy, near-machine-precision satisfaction of the interface jump conditions, and enhanced robustness.
\end{abstract}
\section{Introduction}

Elliptic interface problems describe diffusion in heterogeneous media whose
material properties change across internal interfaces, and their applications are widely found in different scientific fields. 
Let a smooth, closed interface $\Gamma$ divide a bounded domain
$\Omega\subset\mathbb R^2$ into $\Omega^-$ and $\Omega^+$. In this paper, we consider a unified interface neural network (UINN) approach for the elliptic interface problems with governing PDEs, boundary conditions, and jump conditions  as follows:
\begin{subequations}\label{eq:model}
\begin{align}
 -\nabla\!\cdot(\beta\nabla u)&=f
 &&\text{in }\Omega^-\cup\Omega^+,\label{eq:pde}\\
 u&=g &&\text{on }\partial\Omega,\label{eq:boundary}\\
 [u]&=0,\qquad [\beta\partial_{\bn}u]=0
 &&\text{on }\Gamma.\label{eq:transmission}
\end{align}
\end{subequations}
Here $\beta|_{\Omega^s}=\beta^s>0$, $s=\pm$, is constant in each phase,
$\bn$ is the unit normal directed from $\Omega^-$ to $\Omega^+$,
and $[v]=v^-|_\Gamma-v^+|_\Gamma$. Both one-sided derivatives use
this same normal. \eqref{eq:transmission} prescribes the solution and flux jump conditions across the interface.

This type of interface problem has been extensively studied using classical numerical methods, including those defined on meshes that do not fit to the interface. One broad approach incorporates the interface jump conditions through modifications to the discrete scheme or variational formulation, as in cut finite element methods (CutFEM) \cite{2002HansboHansbo, 2016BurmanGuzmanSanchezSarkisHcontrast} and in immersed interface methods (IIM) \cite{1994LevequeLi,2003DumettKeener}. Another approach adapts the approximation space to the interface jump conditions through specially constructed basis functions, as in multiscale finite element methods (MsFEM) \cite{1997HouWu,2010ChuGrahamHou} and extended finite element methods (XFEM) \cite{2006VaughanSmithChopp,2016WangXiaoXu}. Immersed finite element (IFE) methods~\cite{2003LiLinWu, 2004LiLinLinRogers, 2010HeLinLin, 2015LinLinZhang} follow the latter approach: on elements intersected by the interface, they employ piecewise polynomial basis functions that incorporate the jump conditions into the local trial space, thereby enabling interface problems to be solved on interface-unfitted meshes.

Alongside these classical discretization strategies, neural-network approximations offer a flexible and mesh-free way to represent piecewise smooth solutions, and have increasingly been explored for interface problems. Their representational flexibility naturally accommodates phase-dependent solution structures, while the prescribed interface conditions still need to be incorporated appropriately. Physics-informed neural networks (PINNs)~\cite{Raissi2019} provide a natural framework for this purpose by incorporating the governing equations and interface conditions into the training objective. Domain-decomposed formulations such as XPINNs~\cite{XPINN} and interfaced neural networks (INNs)~\cite{LuINN2022} further exploit the flexibility of neural networks to represent distinct solution branches in different subdomains. However, representing the two phases separately does not by itself produce a trial space satisfying the jump conditions, which are typically retained as interface-related residual terms in the loss and balanced against the PDE and boundary residuals.

Beyond subdomain-wise (phase-wise) representations, interface information can also be incorporated directly into the neural trial representation. Interface PINNs (I-PINNs) ~\cite{IPINN2024}, the discontinuity-capturing shallow neural network~\cite{cuspshallow22}, and cusp-capturing PINNs~\cite{cuspPINNell23} use phase-dependent activations, subdomain indicators, and level-set-based features, respectively, to represent piecewise solutions and partially enforce interfacial physics such as discontinuities in the unweighted flux jump across the interface. The immersed interface neural network (IINN)~\cite{Zhang2025} incorporates interface structure through a learned lifting that separates the irregular contribution from a smoother remainder. These constructions demonstrate different ways to incorporate interface information or partial interfacial physics into neural approximations, but satisfaction of the prescribed jump conditions remains dependent on interface-residual training.

Hard-constraint and projection-based approaches share the goal of reducing reliance on interface residuals. For interface optimal control problems, the construction in \cite{LaiHardConstraint2025} embeds the Dirichlet data and solution jump into the ansatz when the required auxiliary functions are available analytically, while retaining the flux condition in the loss. The deep projection method~\cite{HeDeepProjection2026} enforces constraints derived from local finite-difference equations by projecting neural predictions within a domain-decomposition procedure. These developments provide precedents for incorporating constraints directly into neural predictions, which motivates us to develop a \emph{unified neural trial representation that enforces both solution and flux jump conditions by construction}.

We therefore propose the Unified Interface Neural Network (UINN) for the elliptic interface problem considered here. UINN combines a single shared network representing both phases with a projection--correction construction that enforces jump conditions along the smooth interface for every choice of trainable parameters.
The shared network with a phase indicator is related to the interface-aware architectures discussed above, while the projection--correction construction makes the prescribed interface jump conditions intrinsic to the resulting trial space. This construction also aligns with the trial-space principle of IFE methods, particularly higher-order methods that incorporate jump conditions on actual interfaces~\cite{Adjerid2025, 2019GuoLin, 2025ZZML, 2026LinLinZhang, 2018ZhuangGuo}. The structural alignment lies in embedding interfacial physics into the approximation space, realized here through a unified neural representation.

This representation-level enforcement leads to interface-constraint-free optimization. In standard formulations, the solution and flux jump conditions remain training objectives that must be balanced against the PDE and boundary residuals. Balancing multiple loss terms is a recognized challenge in
PINN training due to possible imbalances in their
backpropagated gradients~\cite{WangTengPerdikaris2021}. In contrast, UINN satisfies these conditions by construction, leaving only the PDE residuals and exterior boundary mismatch to be minimized. Interface matching is therefore no longer a competing training task, and we do not need an interface penalty.
The main contributions and their significance are summarized as follows:
\begin{itemize}
\item \textbf{Unified neural approximation with exact jump enforcement.}
We develop a projection--correction construction that uses a single shared network to represent the solution on both sides of the interface and enforces both the solution and flux jump conditions exactly for any network parameters. Thus, the jump conditions are satisfied by construction rather than learned during training.

\item \textbf{Optimization without interface penalties.}
Because the jump conditions are enforced by construction of the network architecture, UINN removes the corresponding interface penalty terms from the loss. Training therefore focuses on the PDE residuals and boundary conditions, without tuning interface penalty weights or balancing jump errors against PDE errors.

\item \textbf{Bridging classical and neural approximation for interface problems.}
UINN carries over a key design principle of classical methods such as IFE to neural approximation: build interfacial physics into the approximation space rather than into the numerical scheme. The local
projection--correction mechanism provides a common realization of this
principle for different smooth interface geometries.

\item \textbf{Robust performance in more challenging interface regimes.}
By enforcing the jump conditions directly in the neural trial space (network architecture), UINN avoids the sensitivity and competition introduced by interface penalty weights. This may contribute to more robust training across different geometries and coefficient contrasts under the tested configurations, while achieving comparable or better accuracy than the standard interface-penalty-based PINNs at comparable trainable parameter counts. The high-contrast tests particularly highlight this robustness.
\end{itemize}
The remainder of the paper is organized as follows.
Section~\ref{sec:method} presents the UINN construction,
Section~\ref{sec:formulations} describes the numerical formulations
and algorithms, Section~\ref{sec:numerics} reports the numerical experiments and
results, and Section~\ref{sec:conclusion} provides a brief conclusion and discussion of the paper.

\section{UINN Architecture}
\label{sec:method}
The UINN architecture is built around a projection--correction process that transforms a unified neural representation into an approximation that satisfies the interface jump conditions. The projection stage enforces the prescribed interface jump conditions on the interface data, while the correction stage incorporates these projected quantities into the neural field. This construction embeds the interface conditions directly into the approximation space, rather than enforcing them through additional loss penalties.

\subsection{Unified branch representation and interface projection}
\label{sec:projection}

A single smooth network $N_\theta(\bx,\zeta_s)$, where
$s\in\{-,+\}$ indexes the phase and
$\zeta_-=-1$, $\zeta_+=+1$ are the corresponding phase-indicator inputs,
defines the two raw branch fields
\begin{equation}
 \widehat u_\theta^s(\bx)=N_\theta(\bx,\zeta_s),\qquad
 \widehat g_\theta^s=\widehat u_\theta^s|_\Gamma,\qquad
 \widehat a_\theta^s=\partial_{\bn}\widehat u_\theta^s|_\Gamma.
 \label{eq:raw}
\end{equation}
Spatial derivatives are taken with $\zeta_s$ fixed. The branches share all
trainable parameters, but their values and derivatives on $\Gamma$
are initially unconstrained.

Coefficient-dependent weighting of interface quantities is widely used in robust discretizations of heterogeneous diffusion problems, including Nitsche-type methods \cite{Annavarapu2012,2016BurmanGuzmanSanchezSarkisHcontrast}. Motivated by this general principle, we adopt a diffusion-scaled least-squares metric for the interface-trace projection, measuring the trace correction on each side in the form \(\beta^s(z-\widehat g_\theta^s)\). Thus, at each \(y\in\Gamma\), we define the projected common trace as
\begin{equation}\label{eq:value-projection}
g_\theta^*(y)
=
\arg\min_{z\in\mathbb{R}}
\frac12
\sum_{s\in\{-,+\}}
\left|\beta^s\bigl(z-\widehat g_\theta^s(y)\bigr)\right|^2
=
\frac{(\beta^-)^2\widehat g_\theta^-(y)
      +(\beta^+)^2\widehat g_\theta^+(y)}
     {D},
\qquad
D=(\beta^-)^2+(\beta^+)^2.
\end{equation}
The projected normal derivatives minimize the unweighted Euclidean
distance to $(\widehat a_\theta^-,\widehat a_\theta^+)$ subject to
$\beta^-a^- -\beta^+a^+=0$. With
$r_{q,\theta}=\beta^-\widehat a_\theta^- -\beta^+\widehat a_\theta^+$,
the solution is
\begin{equation}
 a_\theta^{-,*}=\widehat a_\theta^- -\frac{\beta^-}{D}r_{q,\theta},
 \qquad
 a_\theta^{+,*}=\widehat a_\theta^+ +\frac{\beta^+}{D}r_{q,\theta}.
 \label{eq:flux-projection}
\end{equation}
The use of different metrics for the two projections is a design choice,
motivated by the different roles of the trace and normal-derivative
corrections.
The cutoff function in the trace correction
(see \eqref{eq:band}, \eqref{eq:cutoff}, and \eqref{eq:corrected} below)
has second derivatives that scale as $\delta^{-2}$,
so the correction can contribute terms scaling as
$\beta^s b_\theta^s/\delta^2$ to the PDE residual.
Here $b_\theta^s=g_\theta^*-\widehat g_\theta^s$
is the trace adjustment defined in \eqref{eq:discrepancies} below.
This motivates weighting the trace corrections by $\beta^s$.
For the normal derivatives, we instead use the Euclidean projection
to minimize the total adjustment to the raw predictions subject to
flux continuity.
Although these choices are not necessarily optimal, exact satisfaction
of both jump conditions follows from the projection constraints
for any positive coefficients $\beta^\pm$.
Neither projection introduces an additional trainable parameter.

\subsection{A geometry-independent local correction}
\label{sec:extension}

Let $\phi$ be a signed defining coordinate for $\Gamma$, negative in
$\Omega^-$ and positive in $\Omega^+$, and let
\begin{equation}
 \Gamma=\{\phi=0\},\qquad
 \U_\delta=\{\bx:|\phi(\bx)|<\delta\},\qquad
 \overline{\U_\delta}\subset\Omega.
 \label{eq:band}
\end{equation}
On this neighborhood choose a smooth retraction
$\cp:\U_\delta\to\Gamma$ satisfying $\cp(\by)=\by$ for
$\by\in\Gamma$. Write
$\kappa_\phi(\by)=\partial_{\bn}\phi(\by)>0$.
The maps are assumed sufficiently smooth for the corrected field to have
two phase-wise spatial derivatives; for example, $C^3$ geometric data
and a smooth network suffice. A signed-distance coordinate and its
closest-point projection are one admissible choice, but are not required.
In particular, $\delta$ is a width in the coordinate $\phi$, rather than
necessarily a Euclidean distance.

For $s=\pm$, define the discrepancies on $\Gamma$ and their explicitly composed
extension by
\begin{equation}
 b_\theta^s=g_\theta^*-\widehat g_\theta^s,\qquad
 c_\theta^s=a_\theta^{s,*}-\widehat a_\theta^s,\qquad
 B_\theta^s(\bx)=(b_\theta^s\circ\cp)(\bx)
                  =b_\theta^s\bigl(\cp(\bx)\bigr).
 \label{eq:discrepancies}
\end{equation}
The coefficient of the first-order extension is another interface
function,
\begin{equation}
 k_\theta^s(\by)=
 \frac{c_\theta^s(\by)-\partial_{\bn}B_\theta^s(\by)}
      {\kappa_\phi(\by)},\qquad \by\in\Gamma.
 \label{eq:extension-coefficient}
\end{equation}
In this formula, the derivative acts on the full compound function
$B_\theta^s=b_\theta^s\circ\cp$. Equivalently,
\begin{equation}
 \partial_{\bn}B_\theta^s(\by)
 =\nabla_\Gamma b_\theta^s(\by)\cdot
       \bigl(D\cp(\by)\bn(\by)\bigr).
 \label{eq:chain-rule}
\end{equation}
Here $D\cp$ is the Jacobian of the retraction, and
$\nabla_\Gamma b=(I-\bn\bn^T)\nabla\widetilde b|_\Gamma$ is the
surface gradient, with $\widetilde b$ any smooth extension of $b$.
Because $D\cp\bn$ is tangent to $\Gamma$, the right-hand side does
not depend on that auxiliary extension.

To localize the correction to $\mathcal U_\delta$, we use
the compactly supported cutoff
\begin{equation}
 \eta_\delta(t)=\left(1-\frac{|t|}{\delta}\right)_+^4
                   \left(1+4\frac{|t|}{\delta}\right),
 \qquad z_+=\max(z,0).
 \label{eq:cutoff}
\end{equation}
It is $C^2$, equals one at zero, has zero first derivative there,
and vanishes for $|t|\geq\delta$.
With this cutoff, the \emph{corrected neural network solutions}
(i.e., UINN solutions) are defined as
\begin{equation}
 u_\theta^s(\bx)=\widehat u_\theta^s(\bx)
 +\eta_\delta\bigl(\phi(\bx)\bigr)
 \left[b_\theta^s\bigl(\cp(\bx)\bigr)
       +\phi(\bx)k_\theta^s\bigl(\cp(\bx)\bigr)\right]
 \quad\text{in }\U_\delta,
 \label{eq:corrected}
\end{equation}
and $u_\theta^s=\widehat u_\theta^s$ outside the neighborhood.
The piecewise physical field uses $u_\theta^s$ in $\Omega^s$.

The following proposition verifies the key property of UINN:
both interface jump conditions are satisfied by construction,
independently of the network parameters and training process.
\begin{proposition}\ignorespaces
For every parameter vector $\theta$, the correction
\eqref{eq:corrected} has interface trace $g_\theta^*$ and one-sided
normal derivative $a_\theta^{s,*}$. Hence it satisfies both homogeneous
jump conditions in \eqref{eq:transmission} exactly in exact
arithmetic, for any admissible pair $(\phi,\cp)$.
\end{proposition}
\begin{proof}
Since $\phi=0$ and $\cp$ is the identity on $\Gamma$,
$u_\theta^s|_\Gamma=\widehat g_\theta^s+b_\theta^s=g_\theta^*$.
The product and chain rules, with $\eta_\delta'(0)=0$, give
\[
 \partial_{\bn}u_\theta^s|_\Gamma
 =\widehat a_\theta^s+\partial_{\bn}B_\theta^s
       +\kappa_\phi k_\theta^s
 =\widehat a_\theta^s+c_\theta^s=a_\theta^{s,*}.
\]
Equation~\eqref{eq:flux-projection} imposes the common physical flux.
\end{proof}
For a signed-distance coordinate with closest-point retraction,
$\kappa_\phi=1$ and $D\cp\bn=0$, so that
\eqref{eq:extension-coefficient} reduces to $k_\theta^s=c_\theta^s$.
For a general retraction $\cp$, the chain-rule term must be retained.
The same correction formula \eqref{eq:corrected} applies to \emph{different
smooth interfaces}, with only the corresponding geometric maps changing.
As established by the preceding proposition, the correction enforces
the jump conditions \eqref{eq:transmission} exactly, while the overall
solution accuracy and training performance still depend on the
approximation and optimization of the PDE residual.

The correction is applied only in a narrow neighborhood $\mathcal U_\delta$ that lies strictly inside $\Omega$, so that the geometric maps remain smooth and the correction does not interfere with the exterior boundary condition. The width $\delta$ controls the spatial extent of the correction and also affects its derivatives: from the scaled cutoff in \eqref{eq:cutoff}, its first- and second-order derivatives introduce factors of $\delta^{-1}$ and $\delta^{-2}$, respectively. Thus, choosing $\delta$ too small can amplify derivatives in the corrected field and make the PDE residual harder to optimize.

\section{Numerical Formulations and Algorithms }\label{sec:formulations}
\subsection{Residual minimization and training algorithm}
\label{sec:training}

For $s=\pm$, let $X^s\subset\Omega^s$ be bulk collocation sets,
$B^s\subset\U_\delta\cap\Omega^s$ additional half-band sets, and
$X_b\subset\partial\Omega$ boundary points. The additional half-band samples on $B^s$ are used to better resolve and control the PDE residual in the region where the projection--correction modifies the trial field. For a finite set $X$,
$$\mean{v}{X}=|X|^{-1}\sum_{\bx\in X}v(\bx).$$ Define
\begin{equation}
 \R_\theta^s=-\beta^s\Delta u_\theta^s-f,\qquad
 F_s=\left(\frac{1}{|\Omega^s|}\int_{\Omega^s}|f|^2\,d\bx\right)^{1/2},
 \qquad G=\max_{\partial\Omega}|g|.
 \label{eq:scales}
\end{equation}
All scales are positive for the tests below. 
In a problem with vanishing
data, positive reference scales would instead be prescribed. 
{For a phase-wise field $w$, let
$R_w^s=-\beta^s\Delta w^s-f$.}
We define the common PDE-and-boundary residual functional by
\begin{equation}
\mathcal L_{\mathrm{pb}}(w)
=
\frac12\sum_{s\in\{-,+\}}
\left\{
\mean{\left(R_w^s/F_s\right)^2}{X^s}
+
\mean{\left(R_w^s/F_s\right)^2}{B^s}
\right\}
+
\lambda_b
\mean{\left((w-g)/G\right)^2}{X_b}.
\end{equation}
The loss function of UINN is therefore
\begin{equation}
 \mathcal L_U(\theta)=\mathcal L_{\mathrm{pb}}(u_\theta).
  \label{eq:unified-loss}
\end{equation}
In \(\mathcal L_{\mathrm{pb}}\), the bulk and half-band PDE residuals use the same phase-wise normalization \(F_s\), but their averages over \(X^s\) and \(B^s\) enter the loss separately.
Thus, the bulk and near-interface residuals are measured on the same phase-wise scale, while the half-band contributes its own PDE-residual term rather than being absorbed into the bulk average. Training uses only the prescribed source and boundary data. Since the jump conditions are already enforced by the corrected trial field, no interface penalty terms or separate interface-loss samples are required. For a collocation point $\bx$ inside the correction band, the projected interface data used in the correction are evaluated at the corresponding interface point $\cp(\bx)\in\Gamma$. The complete UINN training and evaluation procedure is summarized in
Algorithm~\ref{alg:unified}.

\begin{algorithm}[!htbp]
\caption{UINN approach for an elliptic interface problem}
\label{alg:unified}
\small

\textbf{Data:}
Subdomains $\Omega^\pm$, coefficients $\beta^\pm$, source $f$, and
boundary data $g$; geometric maps $\phi$ and $\cp$, unit normal $\bn$,
and correction width $\delta$; shared network $N_\theta(\bx,\zeta_s)$, $s=\pm $;
collocation-set sizes, residual scales $F_\pm$ and $G$, boundary weight
$\lambda_b$, and Adam/L-BFGS training parameters.
\par\smallskip

\textbf{Result:}
A trained corrected field $u_{\theta^*}$ satisfying both jump conditions
by construction.
\par\medskip

\textbf{Step 1. Initialization.}
Initialize the parameters $\theta$ of the shared network
$N_\theta(\bx,\zeta_s)$. Sample $X^\pm$ uniformly with respect to area in $\Omega^\pm$,
sample $B^\pm$ uniformly with respect to area in
$\U_\delta\cap\Omega^\pm$, and sample $X_b$ uniformly along
$\partial\Omega$.\par\smallskip

\textbf{Step 2. Adam training.}
For each Adam update:
\begin{list}{}{\setlength{\leftmargin}{1.4em}
\setlength{\itemsep}{2pt}\setlength{\topsep}{2pt}}

\item Refresh $X^\pm$, $B^\pm$, and $X_b$ at the prescribed
resampling interval.

\item For each phase collocation point $\bx$, evaluate the raw field
$\widehat u_\theta^s(\bx)=N_\theta(\bx,\zeta_s)$. If
$\bx\in\U_\delta$, set $\by=\cp(\bx)$, evaluate both raw traces
$\widehat g_\theta^\pm(\by)$ and normal derivatives
$\widehat a_\theta^\pm(\by)$, and project them by
\eqref{eq:value-projection}--\eqref{eq:flux-projection} to obtain
$g_\theta^*(\by)$ and $a_\theta^{\pm,*}(\by)$.

\item Use these projected interface data to form
$b_\theta^\pm$, $c_\theta^\pm$, and $k_\theta^\pm$ from
\eqref{eq:discrepancies}--\eqref{eq:extension-coefficient}, and evaluate
the corrected field $u_\theta^s(\bx)$ from \eqref{eq:corrected}.
Outside $\U_\delta$, set
$u_\theta^s(\bx)=\widehat u_\theta^s(\bx)$.

\item Compute
\[
\R_\theta^s(\bx)
=
-\beta^s\Delta u_\theta^s(\bx)-f(\bx),\quad s=\pm,
\]
on $X^\pm$ and $B^\pm$, together with the boundary mismatch on $X_b$,
and assemble $\mathcal L_U(\theta)$ from
\eqref{eq:unified-loss}. Differentiate through the complete corrected
field, backpropagate, and update $\theta$.

\end{list}

\textbf{Step 3. L-BFGS refinement.}
Draw fresh sets $X^\pm$, $B^\pm$, and $X_b$ and keep them fixed
throughout L-BFGS. Repeatedly evaluate the same corrected field and loss
as in Step~2, and minimize $\mathcal L_U$ using the prescribed
strong-Wolfe L-BFGS procedure. Retain the monitored parameter vector with
the smallest training loss and stop according to the prescribed
iteration budget, loss target, or optimizer stopping criterion.
\par\smallskip

\textbf{Step 4. Independent evaluation.}
Restore the retained parameters $\theta^*$. On independently generated
solution, residual, boundary, and interface evaluation points, compute
the solution errors, PDE residuals, boundary errors, and interface jump
and interface-accuracy measures defined in Section~\ref{sec:numerics}.
\end{algorithm}

The PDE residual is evaluated from the corrected field $u_\theta^s$, not
from the raw branch $\widehat u_\theta^s$. Therefore, the correction terms
and their dependence on the projected interface data are differentiated
together with $u_\theta^s$ when computing $\Delta u_\theta^s$. Since the correction coefficient $c_\theta^s$ in
\eqref{eq:discrepancies} contains the normal derivative
$\widehat a_\theta^s=\partial_{\bn}\widehat u_\theta^s$ defined in
\eqref{eq:raw}, evaluating the PDE residual of the corrected field
generally requires derivatives of $\widehat u_\theta^s$ up to third
order. This higher differentiation order
is one source of the additional computational cost of the current UINN
implementation.

\subsection{Standard piecewise PINN benchmark}
\label{sec:standard-method}

The standard benchmark uses two independent networks
$v_\psi^-(\bx)$ and $v_\psi^+(\bx)$ to approximate the
solutions in $\Omega^-$ and $\Omega^+$, respectively. Unlike the UINN,
no projection--correction is applied to these two fields. Their PDE and
boundary contributions are evaluated using the same functional
$\mathcal L_{\mathrm{pb}}$ defined above, with $v_\psi$ replacing
$u_\theta$ and with the same phase-wise and half-band residual terms.

To couple the two independently represented phase solutions, we introduce
an additional interface collocation set $X_\Gamma$. At each point in
$X_\Gamma$, both networks and their normal derivatives are evaluated,
and the solution and physical-flux jumps are penalized through
\begin{align}
 \mathcal L_S(\psi)
 &=
 \mathcal L_{\mathrm{pb}}(v_\psi)
 +\lambda_u
 \mean{([v_\psi]/U_\Gamma)^2}{X_\Gamma}
 +\lambda_q
 \mean{([\beta\partial_{\bn}v_\psi]/Q_\Gamma)^2}{X_\Gamma},
 \label{eq:standard-loss}\\
 U_\Gamma&=G,\qquad
 Q_\Gamma=\frac{F_-|\Omega^-|}{|\Gamma|}.
 \label{eq:standard-scales}
\end{align}
Here $U_\Gamma=G$ uses the boundary-data scale for the solution jump,
whereas $Q_\Gamma$ provides a characteristic interface-flux scale based
on the source strength in the enclosed phase $\Omega^-$. These scales
depend only on the prescribed problem data and geometry and do not use
exact interface values or derivatives.

The points in $X_\Gamma$ are sampled independently and stratified
uniformly in arc length along $\Gamma$. A separate random stream is used
for $X_\Gamma$ so that changing the number of interface points does not
alter the phase-wise, half-band, or boundary collocation sets. The
half-band sets $B^\pm$ continue to contribute PDE residual terms, as in
the UINN, so that the two methods use the same additional near-interface
PDE sampling; for the standard PINN, however, these points are not
associated with any correction of the trial field. The complete standard PINN training and evaluation procedure is summarized
in Algorithm~\ref{alg:standard}.

\begin{algorithm}[!htbp]\captionsetup{font={small},labelfont={bf}}
\caption{Standard piecewise PINN benchmark}
\label{alg:standard}
\small

\textbf{Data:}
Subdomains $\Omega^\pm$, coefficients $\beta^\pm$, source $f$, and
boundary data $g$; two independent networks
$v_\psi^-(\bx)$ and $v_\psi^+(\bx)$; collocation-set sizes,
residual scales $F_\pm$, $G$, $U_\Gamma$, and $Q_\Gamma$; weights
$\lambda_b$, $\lambda_u$, and $\lambda_q$; and Adam/L-BFGS training
parameters.
\par\smallskip

\textbf{Result:}
A trained piecewise field $v_{\psi^*}$.
\par\medskip

\textbf{Step 1. Initialization.}
Initialize the parameters of the two independent networks.
Sample $X^\pm$ uniformly with respect to area in $\Omega^\pm$,
sample $B^\pm$ uniformly with respect to area in
$\U_\delta\cap\Omega^\pm$, and sample $X_b$ uniformly along
$\partial\Omega$. Independently sample $N_\Gamma$ points in
$X_\Gamma$ uniformly with respect to arc length on $\Gamma$.
Compute $U_\Gamma$ and $Q_\Gamma$ from
\eqref{eq:standard-scales}.
\par\smallskip

\textbf{Step 2. Adam training.}
For each Adam update:
\begin{list}{}{\setlength{\leftmargin}{1.4em}
\setlength{\itemsep}{2pt}\setlength{\topsep}{2pt}}

\item Regenerate the collocation sets at the prescribed resampling
interval, using a separate random stream for $X_\Gamma$.

\item For $s=\pm$, evaluate $v_\psi^s$ on $X^s$ and $B^s$ and
compute the phase-wise PDE residual
\[
R_{v_\psi}^s
=
-\beta^s\Delta v_\psi^s-f.
\]
Evaluate the boundary mismatch on $X_b$.

\item At each point in $X_\Gamma$, evaluate both network branches and
their normal derivatives to compute the solution jump
$[v_\psi]$ and flux jump
$[\beta\partial_{\bn}v_\psi]$.

\item Assemble $\mathcal L_S(\psi)$ from
\eqref{eq:standard-loss}, backpropagate, and update $\psi$.

\end{list}

\textbf{Step 3. L-BFGS refinement.}
Generate fresh sets $X^\pm$, $B^\pm$, $X_b$, and $X_\Gamma$ and keep
them fixed throughout L-BFGS. Repeatedly evaluate the same PDE,
boundary, and interface residuals as in Step~2 and minimize
$\mathcal L_S$ using the prescribed strong-Wolfe L-BFGS procedure.
Retain the monitored parameter vector with the smallest training loss and
stop according to the prescribed iteration budget, loss target, or
optimizer stopping criterion.
\par\smallskip

\textbf{Step 4. Independent evaluation.}
Restore the retained parameters $\psi^*$. On independently
generated solution, residual, boundary, and interface evaluation
points, compute the solution errors, PDE residuals, boundary errors,
interface jump errors, and interface-accuracy measures defined in
Section~\ref{sec:numerics}.
\end{algorithm}

\FloatBarrier

\section{Numerical Tests}
\label{sec:numerics}

We consider two manufactured interface problems on
$\Omega=(-1,1)^2$. {The circular and star-shaped interface experiments in
Sections~\ref{sec:circle} and~\ref{sec:star} use modestly large contrast
$(\beta^-,\beta^+)=(1,100)$; Section~\ref{sec:star1000} adds the
star-shape interface test at high-contrast coefficient $(\beta^-,\beta^+)=(1,1000)$.}
The circle test follows the radial manufactured-solution construction
in the immersed finite element study~\cite{2025ZZML}; the star-shaped interface test uses
a polynomial manufactured solution with the same homogeneous
jump conditions. We use $\delta=0.2$ for circular interface tests and $\delta=0.15$ in star-shaped interface tests, as indicated in Table \ref{tab:common-setup}. Unless stated otherwise, we set \(\lambda_b=10\) for both UINN and standard PINNs and \(\lambda_u=\lambda_q=10\) for the standard PINN benchmarks.

\paragraph{Network architectures and optimization setup.}
The UINN uses a single shared network with three hidden tanh layers of
width $40$, giving $3,481$ trainable parameters. The primary standard
PINN uses two independent width-$28$ branches, giving $3,474$ trainable
parameters in total and providing an approximately equal-parameter
comparison. We additionally consider a standard PINN with two width-$40$
branches, giving $6,882$ parameters in total, as an equal-hidden-width
comparison.

All networks use Xavier-normal initialization, zero initial biases, unit
output gain, and double-precision arithmetic.
Each of the three experiment groups
(Sections~\ref{sec:circle}--\ref{sec:star1000})
uses $2,000$ Adam updates with a
\emph{cosine learning-rate} decay from $10^{-3}$ to $10^{-4}$.
The collocation points are regenerated every $100$ Adam
updates.

The subsequent L-BFGS refinement uses a fresh collocation set
that is held fixed throughout the refinement.
The L-BFGS iteration budgets are summarized in
Table~\ref{tab:common-setup}, with the same Adam and L-BFGS
iteration budgets used by UINN and standard PINNs within each
experiment group.

L-BFGS refinement is performed on a fixed collocation set using a
strong-Wolfe line search. The loss is monitored every $25$ iterations,
with at most $50$ loss-and-gradient evaluations allowed over each such
interval, and the iterate attaining the smallest monitored loss is
retained. The stopping thresholds are $10^{-10}$ for the gradient norm
and $10^{-13}$ for successive changes in the parameters or loss.

\paragraph{Sampling and residual scales.}
Collocation points in $X^s$ are sampled uniformly by area over the
entire subdomain $\Omega^s$, without excluding the half-band.
An additional set of collocation points, denoted by $B^s$, is sampled
uniformly by area from the half-band $\U_\delta\cap\Omega^s$.
The boundary points in $X_b$ are sampled uniformly along each edge. For convenience, define
\[
N_I
=
|X^-|+|X^+|+|B^-|+|B^+|,
\qquad
N_b=|X_b|,
\]
where $N_I$ is the total number of PDE collocation points used in the
interior, including the additional half-band points, and $N_b$ is the
total number of outer-boundary collocation points.

For each geometry, the UINN and the approximately equal-parameter
standard PINN configuration used for the main comparison use exactly
the same point coordinates for the phase-wise PDE collocation sets
$X^\pm$, the additional half-band sets $B^\pm$, and the boundary set
$X_b$. They are also evaluated on the same independently generated
test points. Only the standard PINN uses an additional interface set
$X_\Gamma$ for the interface-penalty terms. Table~\ref{tab:common-setup} summarizes the correction-band width,
training collocation counts, and optimization budgets for the three
numerical test settings in Sections~\ref{sec:circle}, \ref{sec:star}, and \ref{sec:star1000}.

\begin{table}[!htbp]
\centering
\small
\begin{tabular}{@{}lccc@{}}
\toprule
Quantity
& Circle $(1{,}100)$
& Star $(1{,}100)$
& Star $(1{,}1000)$ \\
\midrule
$\delta$
& $0.20$ & $0.15$ & $0.15$\\
$|X^-|$
& $256$ & $512$ & $512$\\
$|X^+|$
& $768$ & $1024$ & $1024$\\
$|B^-|$
& $128$ & $256$ & $256$\\
$|B^+|$
& $128$ & $256$ & $256$\\
$|X_b|$
& $256$ & $512$ & $512$\\
Adam iterations
& $2000$ & $2000$ & $2000$\\
L-BFGS iterations
& $1600$ & $4000$ & $8000$\\
\bottomrule
\end{tabular}
\caption{Common training setup within each of the three numerical test
settings. The listed phase-wise, half-band, and boundary collocation
counts and optimization budgets are shared by the UINN and all standard
PINN benchmark cases in the corresponding test for different interfaces and $(\beta^-, \beta^+)$.}
\label{tab:common-setup}
\end{table}

\paragraph{Computational cost metrics.}
To compare the computational work of the two approaches, we consider
the number of network evaluations required for one evaluation of the
training loss.

For the UINN, let $N_A$ denote the number of PDE collocation points
that lie in the active correction band $\U_\delta$. Importantly,
$N_A$ does not denote an additional collocation set: it counts all
points from $X^-\cup X^+\cup B^-\cup B^+$ that lie in
$\U_\delta$. Thus, all points in $B^\pm$ contribute to $N_A$, together
with those points in the phase-wise sets $X^\pm$ that also fall inside
the correction band. At each such point, the correction requires
evaluating both raw branches at the corresponding interface point
$\cp(\bx)$. Hence one UINN loss evaluation involves
\(
N_I+N_b+2N_A
\)
network-point evaluations.

For the standard PINN, let
$N_\Gamma=|X_\Gamma|$ denote the number of interface collocation points
used in the jump-penalty terms. Since both independent branches are
evaluated at every interface point, one standard-PINN loss evaluation
involves
\(
N_I+N_b+2N_\Gamma
\)
branch-point evaluations.

In addition to these network-point counts, the two methods differ in
differentiation cost. The UINN PDE residual generally requires
raw-network spatial derivatives up to third order because the correction
already contains first derivatives before the Laplacian is applied,
whereas the standard PINN requires derivatives up to second order.
We therefore report the resulting network-point and forward multiply–accumulate (MAC) counts
per loss evaluation, together with the L-BFGS closure counts and measured
training wall time. The detailed computational-cost comparison is
reported for the circular-interface experiment in
Section~\ref{sec:circle}. The corresponding results for the other two cases
are omitted since they lead to the same qualitative cost comparison.

\paragraph{Accuracy measures.}
Let $e=u_h-u$, where $u_h$ is the evaluated network approximation. For $s=\pm$, we use
\begin{equation}
 E_0=\frac{\|e\|_{L^2(\Omega)}}{\|u\|_{L^2(\Omega)}},\qquad
 E_1=\frac{\left(\sum_s\|\nabla e^s\|_{L^2(\Omega^s)}^2\right)^{1/2}}
              {\left(\sum_s\|\nabla u^s\|_{L^2(\Omega^s)}^2\right)^{1/2}},
 \qquad
 E_E=\frac{\left(\sum_s\beta^s\|\nabla e^s\|_{L^2(\Omega^s)}^2\right)^{1/2}}
              {\left(\sum_s\beta^s\|\nabla u^s\|_{L^2(\Omega^s)}^2\right)^{1/2}}.
 \label{eq:errors}
\end{equation}
The phase-restricted versions use the exact norm on that phase as
denominator. 
Let $\mathcal G_h$ denote the $200\times200$ Cartesian cell-center grid
with spacing $h=2/200=0.01$. The norm integrals are approximated by
midpoint quadrature on $\mathcal G_h$ with weight $h^2$, and
\[
e_{h,\infty}
=
\max_{\bx\in\mathcal G_h}|e(\bx)|
\]
denotes the maximum on these points. Solution images use a separate
$201\times201$ grid. All tabulated relative errors are recorded in percentages.
We quantify the solution and flux jump errors across the interface using the following measures:
\begin{equation}
 J_{u,\infty}=\max_{\Gamma_h}|[u_h]|,\quad
 J_{q,\infty}=\max_{\Gamma_h}|[\beta\partial_{\bn}u_h]|,\quad
 J_{u,2}=\operatorname{RMS}_{\Gamma_h}[u_h],\quad
 J_{q,2}=\operatorname{RMS}_{\Gamma_h}[\beta\partial_{\bn}u_h].
 \label{eq:jumps-norms}
\end{equation}
Here $\Gamma_h$ denotes an independent set of $1,021$ interface
evaluation points generated from the angular parameters
\[
\vartheta_j=\frac{2\pi(j+0.371)}{1021},
\qquad j=0,\ldots,1020.
\]
For the circular interface, using the radius $r_0$ defined in
Section~\ref{sec:circle}, these points are
$r_0(\cos\vartheta_j,\sin\vartheta_j)$. For the star-shaped interface,
using the radial function $\rho$ defined in
\eqref{eq:star-interface}, they are
$\rho(\vartheta_j)(\cos\vartheta_j,\sin\vartheta_j)$.
The shift $0.371$ avoids coincidence with the training interface points
and special symmetry locations. The RMS values on $\Gamma_h$ use
arc-length weights; for a polar curve, these weights are proportional to
$\sqrt{\rho^2+(\rho')^2}$.

At the independent interface evaluation points $\Gamma_h$, we compare
the minus-side numerical trace $u_h^-$ with the exact interface value
$u|_\Gamma$, and the minus-side physical normal flux
$\beta^-\partial_{\bn}u_h^-$ with the exact physical normal flux
$q_\Gamma$. We measure the trace error by its RMS and maximum values,
and the flux error by its RMS value and corresponding relative RMS value:\begin{subequations}\label{eq:interface-errors}
\begin{align}
E_{\Gamma,u,2}
&=
\operatorname{RMS}_{\Gamma_h}
\bigl(u_h^- - u|_\Gamma\bigr),
&
E_{\Gamma,u,\infty}
&=
\max_{\Gamma_h}
\bigl|u_h^- - u|_\Gamma\bigr|,
\label{eq:interface-errors-u}
\\
E_{\Gamma,q,2}
&=
\operatorname{RMS}_{\Gamma_h}
\bigl(\beta^-\partial_{\bn}u_h^- - q_\Gamma\bigr),
&
E_{\Gamma,q,\mathrm{rel}}
&=
\frac{E_{\Gamma,q,2}}
{\operatorname{RMS}_{\Gamma_h}(q_\Gamma)},
\label{eq:interface-errors-q}
\end{align}
\end{subequations}
where
$q_\Gamma=\beta^-\partial_{\bn}u^-|_\Gamma
=\beta^+\partial_{\bn}u^+|_\Gamma$
is the exact physical normal flux. Here,
$E_{\Gamma,q,\mathrm{rel}}$ is reported because $q_\Gamma$ varies
significantly along the interface, providing a scale-normalized measure of the
normal-flux error.

For any independent residual evaluation set $Y\subset\Omega^s$, the
reported relative PDE residual is
\[
\left(
\frac{\sum_{\bx\in Y}|\mathcal R_h(\bx)|^2}
     {\sum_{\bx\in Y}|f(\bx)|^2}
\right)^{1/2}.
\]
For half-band residuals, $Y$ is sampled independently in
$\U_\delta\cap\Omega^s$. For all numerical tests, boundary errors are
evaluated independently at $257$ points on each edge together with the
four corners, giving $1,032$ boundary evaluation points in total.

\paragraph{Timing and comparison scope.}
The circular and star-shaped experiments with
$(\beta^-,\beta^+)=(1,100)$ were run on an AMD EPYC 9V74 CPU using
one PyTorch thread, Python~3.12.14, PyTorch~2.14.0+cpu, and double
precision. Within each of the three test settings in
Table~\ref{tab:common-setup}, the UINN and standard PINN use the same
phase-wise, half-band, and boundary collocation points, the same
independent evaluation sets, unit output gain, and the same Adam and
L-BFGS iteration budgets. The standard PINN additionally uses its
interface collocation set $X_\Gamma$.


For the reported timing results, training wall time denotes the total
Adam and L-BFGS time, including collocation resampling during Adam and
the periodic loss evaluations during L-BFGS, but excluding the
independent post-training evaluation, visualization, and data export.
Equal optimizer iteration budgets do not imply equal computational
work because the UINN and standard PINN have different costs per loss
evaluation, as described above.
\subsection{Circular interface}
\label{sec:circle}

Let $R=(x^2+y^2)^{1/2}$, $r_0=\pi/6.28$, $r=R/r_0$, and
$\Gamma=\{R=r_0\}$. The inner disk is $\Omega^-$.
For general diffusion coefficients $(\beta^-,\beta^+)$, the terms $f$ and $g$ can be generated with the
manufactured solution
\begin{equation} \label{eq:circle-exact}
 u^s(x,y)
 =1+\frac{r^5-1}{\beta^s},\quad s=\pm.
 \end{equation}
 Therefore
 \begin{equation*}
 f(x,y)=-\frac{25}{r_0^2}r^3,
 \qquad
 g=1+\frac{r^5-1}{\beta^+}
 \quad\text{on }\partial\Omega.
\end{equation*}

\subsubsection{UINN experiments and results}
\label{sec:circle-unified}

For the circular interface, the interface is given by $R=r_0$, so we
take the signed radial defining coordinate $\phi$ and the corresponding
radial retraction $\cp$ as
\begin{equation}
 \phi(x,y)=R-r_0,\qquad
 \cp(x,y)=\frac{r_0}{R}(x,y)^T.
 \label{eq:circle-geometry}
\end{equation}These maps are used only in the annular correction neighborhood, away
from the origin and the exterior boundary. Here $\kappa_\phi=1$ and
$D\cp\bn=0$, so that $k_\theta^s=c_\theta^s$.

Using the training setup in Table~\ref{tab:common-setup},
Figure~\ref{fig:circle-u-training} shows the collocation distribution
and optimization history, while Figure~\ref{fig:circle-u-solution}
shows the exact solution, UINN approximation, and pointwise absolute
error.

\begin{figure}[!htbp]\centering
\includegraphics[width=\linewidth]{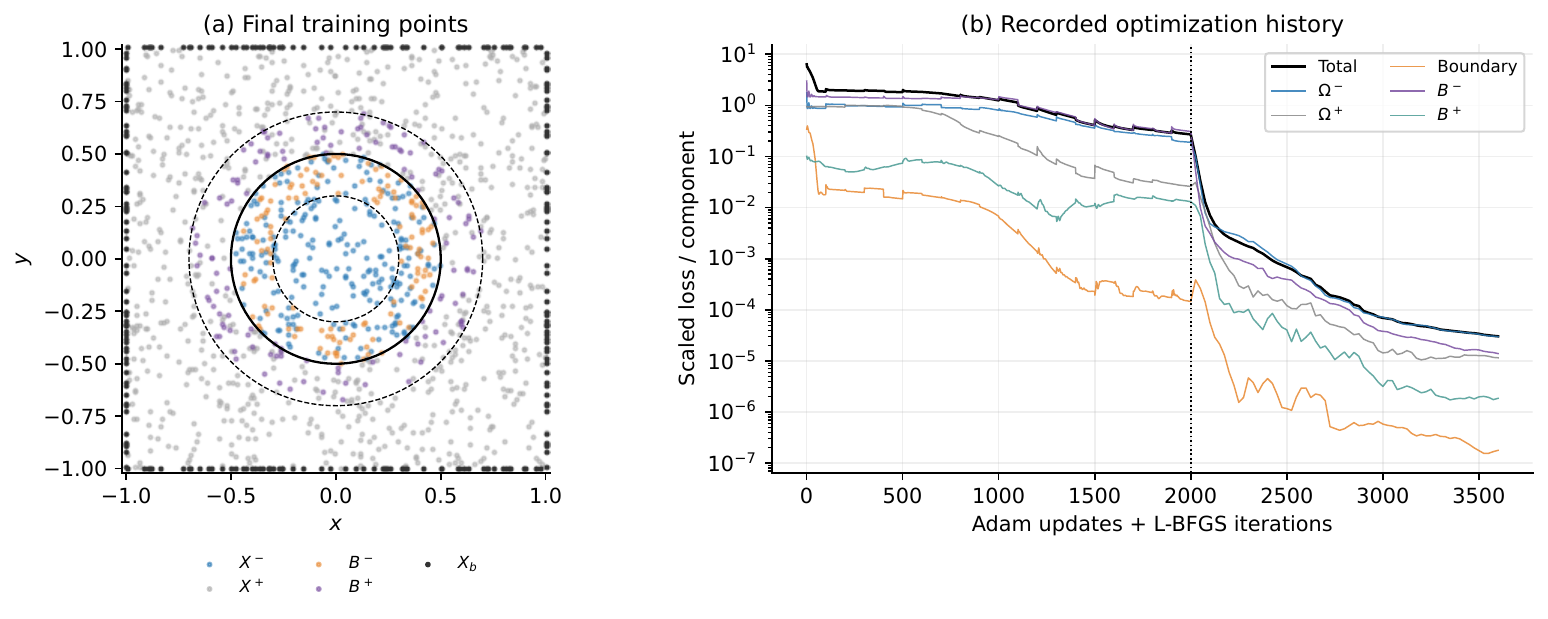}
\caption{UINN training for the circular-interface problem with
$(\beta^-,\beta^+)=(1,100)$. Left: collocation points in
$\Omega^-$ and $\Omega^+$, in the two half-bands, and on the outer
boundary. The solid curve denotes $\Gamma$, and the dashed curves denote
the boundaries of the correction band $|\phi|<\delta$, with
$\delta=0.20$. Right: optimization history of the total loss and its
normalized residual components. The dotted vertical line separates
Adam from L-BFGS.}\label{fig:circle-u-training}
\end{figure}

\begin{figure}[!htbp]\centering
\includegraphics[width=\linewidth]{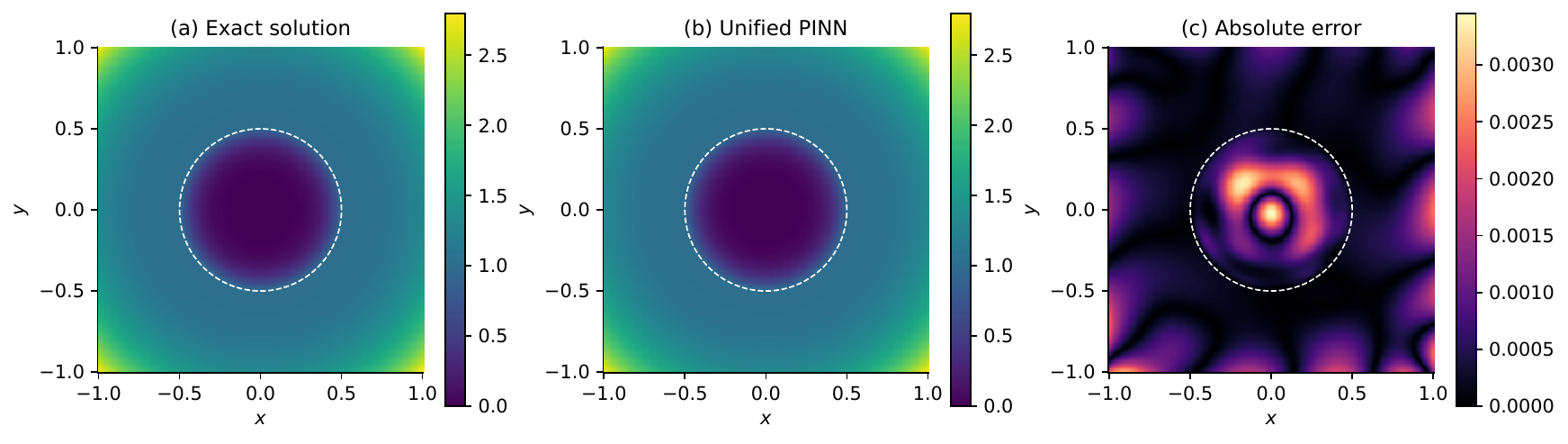}
\caption{Results on visualization grid for the circular interface, $(\beta^-,\beta^+)=(1,100)$: exact
solution (left), UINN solution (middle), and pointwise absolute error (right). The dashed circle marks $\Gamma$.}
\label{fig:circle-u-solution}
\end{figure}
Rows ``U'' of Tables~\ref{tab:circle-s-errors} and
\ref{tab:circle-s-residuals} report the solution errors and
independently evaluated PDE residuals for the circular UINN
experiment. The circular-interface columns of
Table~\ref{tab:uinn-raw-corrected-jumps} compare the raw and
corrected interface jumps. The UINN achieves small solution
errors in both subdomains and satisfies both jump conditions
to nearly machine precision, whereas the raw branches exhibit noticeable
jump errors. These results show that the projection--correction
construction enforces the jump conditions independently of
the residual minimization used to approximate the solution.

\par
\begin{table}[!htbp]
\centering
\footnotesize
\setlength{\tabcolsep}{3pt}
\renewcommand{\arraystretch}{1.13}
\captionsetup{font={small},labelfont={bf}}
\begin{tabular}{@{}lrrrrrr@{}}
\toprule
& \multicolumn{2}{c}{Circular interface}
& \multicolumn{2}{c}{Star-shaped interface}
& \multicolumn{2}{c}{Star-shaped interface} \\
& \multicolumn{2}{c}{$(\beta^-,\beta^+)=(1,100)$}
& \multicolumn{2}{c}{$(\beta^-,\beta^+)=(1,100)$}
& \multicolumn{2}{c}{$(\beta^-,\beta^+)=(1,1000)$} \\
\cmidrule(lr){2-3}\cmidrule(lr){4-5}\cmidrule(l){6-7}
Field & $J_{u,\infty}$ & $J_{q,\infty}$
& $J_{u,\infty}$ & $J_{q,\infty}$ & $J_{u,\infty}$ & $J_{q,\infty}$ \\
\midrule
Raw branches $\widehat u_\theta^\pm$ & $1.940\times10^{-2}$ & $1.493\times10^{-1}$ & $4.533\times10^{-2}$ & $1.266\times10^{0}$ & $2.514\times10^{-2}$ & $1.242\times10^{0}$ \\
Corrected field $u_\theta^\pm$ & $3.109\times10^{-15}$ & $5.329\times10^{-14}$ & $2.487\times10^{-14}$ & $2.043\times10^{-12}$ & $3.020\times10^{-14}$ & $3.524\times10^{-12}$ \\
\bottomrule
\end{tabular}
\caption{Maximum solution and flux jumps for the raw branch fields
$\widehat u_{\theta}^{\pm}$ and the corrected fields
$u_{\theta}^{\pm}$ for the UINN in the three numerical tests,
evaluated at $1,021$ separate post-training interface points in each test.
The three column groups correspond to the circular-interface problem
with $(\beta^-,\beta^+)=(1,100)$ and the star-shaped-interface problems
with $(\beta^-,\beta^+)=(1,100)$ and $(1,1000)$, respectively.}
\label{tab:uinn-raw-corrected-jumps}
\end{table}

{This experiment and the related results verify the intended projection--correction mechanism:
the UINN retains an accurate approximation to the reference solution in the two subdomains while
enforcing both jump conditions to nearly machine precision ({Table~\ref{tab:uinn-raw-corrected-jumps}, circular-interface columns}). The corresponding
raw branch fields do not satisfy the interface conditions, as expected by design.\par}

\begingroup\ignorespaces
\subsubsection{Standard PINN benchmark tests}
\label{sec:circle-standard}

For the circular-interface standard PINN benchmark, C128 is designated as the primary
standard PINN case because its two width-$28$ branches contain $3,474$
trainable parameters in total, closely matching the $3,481$ parameters
of the UINN. C64 and C256 use the same architecture and differ only in
the number of interface collocation points (64 and 256 points respectively), providing a check on the
effect of interface sampling. CW40 keeps $128$ interface points but uses
width $40$ in each branch, providing a comparison
with an equal network width to UINN, but with an increased total parameter count. All four standard cases use the
same phase-wise, half-band, and boundary collocation sets as the UINN
and the same $2,000$ Adam updates and $1,600$ L-BFGS iterations.
Table~\ref{tab:circle-s-cases} summarizes these configurations. 

\emph{Throughout the paper}, the label ``U'' in the tables denotes the UINN case.
\par
\begin{table}[!htbp]
\centering
\small
\renewcommand{\arraystretch}{1.13}
\ignorespaces
\captionsetup{font={small},labelfont={bf}}
\begin{tabular}{@{}lrrrrr@{}}
\toprule
Case & Hidden width & $N_\Gamma$ & Parameters & $\lambda_u$ & $\lambda_q$ \\
\midrule
U & 40 / shared & 0 & 3,481 & -- & -- \\
C128 & 28 / each & 128 & 3,474 & 10 & 10 \\
C64 & 28 / each & 64 & 3,474 & 10 & 10 \\
C256 & 28 / each & 256 & 3,474 & 10 & 10 \\
CW40 & 40 / each & 128 & 6,882 & 10 & 10 \\
\bottomrule
\end{tabular}
\caption{Network, interface sampling, and penalty configurations for the circular-interface problem with
$(\beta^-,\beta^+)=(1,100)$. C128 is the primary approximately
equal-parameter case; C64 and C256 vary the number of interface
collocation points, while CW40 uses the same hidden width as the UINN. All networks have three hidden layers. }
\label{tab:circle-s-cases}
\end{table}
\par
\begin{figure}[!htbp]\captionsetup{font={small},labelfont={bf}}\centering
\includegraphics[width=\linewidth]{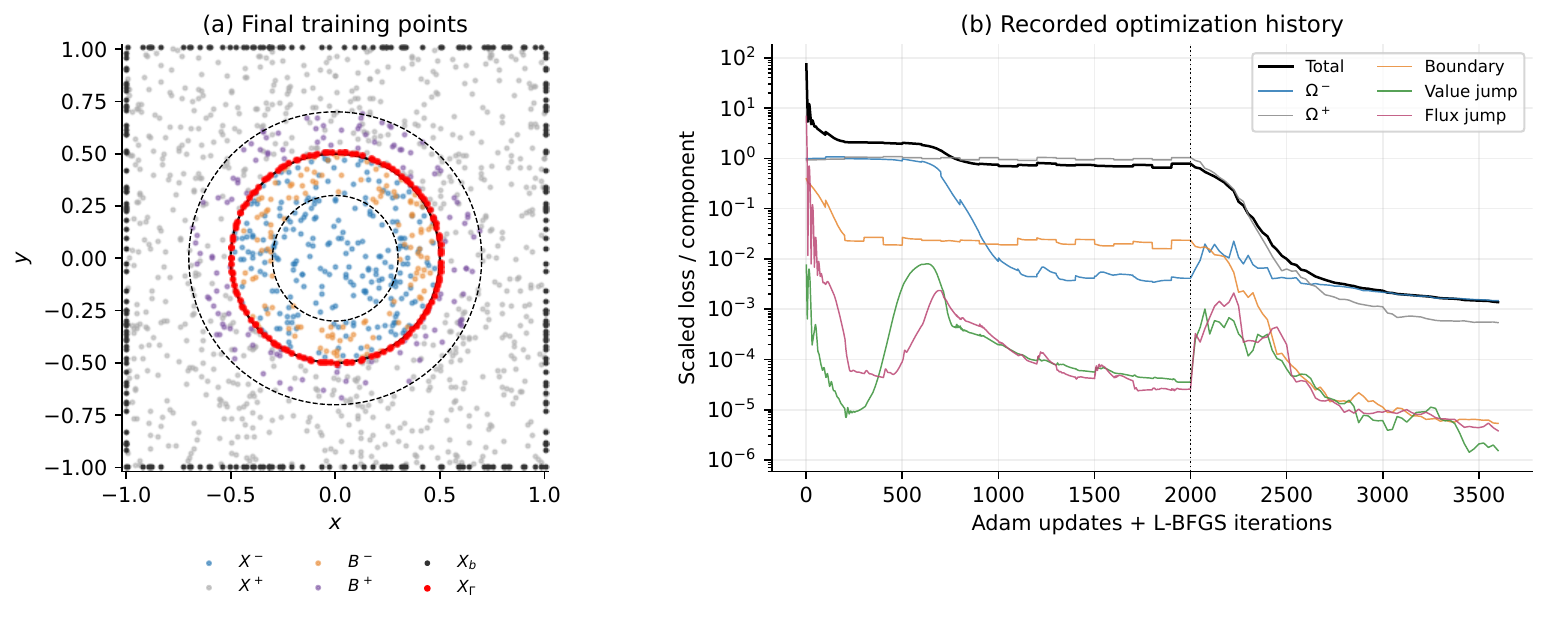}
\caption{Standard PINN case C128 for the circular-interface problem with
$(\beta^-,\beta^+)=(1,100)$. Left: phase-wise, half-band,
boundary, and interface collocation points. Right: training history of
the total loss and its PDE, boundary, value-jump (solution-jump), and flux-jump components.}\label{fig:circle-s-training}
\end{figure}

Figure~\ref{fig:circle-s-training} shows the collocation distribution
and optimization history for the primary PINN C128 case, while
Figure~\ref{fig:circle-s-solution} shows the exact solution, PINN
approximation, and pointwise absolute error.
{As indicated in Tables~\ref{tab:circle-s-cases} and
\ref{tab:circle-s-errors}, the UINN achieves substantially smaller
solution errors than the approximately equal-parameter C128 case,
in the global $L^2$, broken $H^1$, energy errors, and sampled maximum
absolute errors.
Increasing the standard PINN width improves the accuracy, but CW40
still has larger errors than the UINN despite using
nearly twice as many parameters. Increasing the number of interface
collocation points from $64$ to $256$ does not produce a systematic
improvement, suggesting that insufficient interface sampling is not
the main reason for the errors in these standard PINN results.}

\begin{table}[!htbp]
\centering
\small
\renewcommand{\arraystretch}{1.13}
\footnotesize
\setlength{\tabcolsep}{3.5pt}
\ignorespaces
\captionsetup{font={small},labelfont={bf}}
\captionsetup{font={small},labelfont={bf}}
\begin{tabular}{@{}lrrrrrr@{}}
\toprule
Case & $E_0$ (\%) & $E_0^-$ (\%) & $E_0^+$ (\%)
     & $E_1$ (\%) & $E_E$ (\%) & $e_{h,\infty}$ \\\midrule
U & 0.06590 & 0.31453 & 0.04445 & 0.34218 & 0.34641 & $3.456\times10^{-3}$ \\
C128 & {2.04200} & {12.85481} & {0.37345} & {4.90296} & {1.85130} & {$1.225\times10^{-1}$} \\
C64 & {2.09817} & {13.24985} & {0.34614} & {5.23487} & {1.83306} & {$1.262\times10^{-1}$} \\
C256 & {2.13525} & {13.49574} & {0.34085} & {5.10561} & {1.69182} & {$1.274\times10^{-1}$} \\
CW40 & {0.57016} & {3.32306} & {0.23858} & {1.56243} & {0.83943} & {$3.155\times10^{-2}$} \\
\bottomrule
\end{tabular}
\caption{Solution error comparison between the UINN and standard PINN
cases for the circular-interface problem with
$(\beta^-,\beta^+)=(1,100)$.}
\label{tab:circle-s-errors}
\end{table}

\begin{figure}[!htbp]\captionsetup{font={small},labelfont={bf}}\centering
\includegraphics[width=\linewidth]{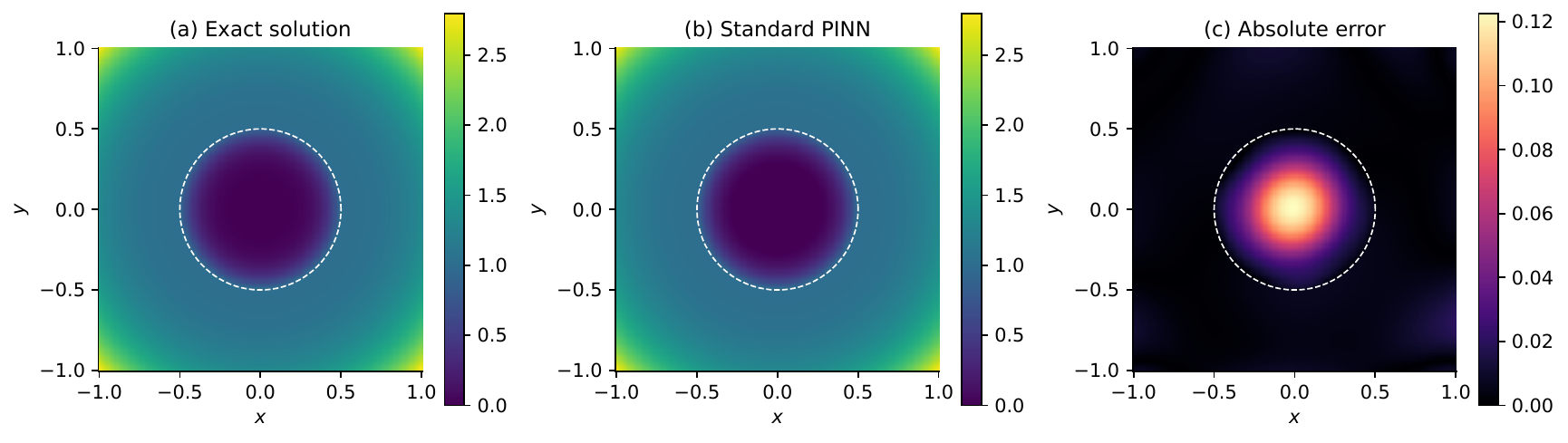}
\caption{Results for the standard PINN case C128 on the circular-interface
problem with $(\beta^-,\beta^+)=(1,100)$: exact solution (left),
standard PINN solution (middle), and pointwise absolute error (right).}\label{fig:circle-s-solution}
\end{figure}

\par
\begin{table}[!htbp]
\centering
\small
\renewcommand{\arraystretch}{1.13}
\ignorespaces
\captionsetup{font={small},labelfont={bf}}
\captionsetup{font={small},labelfont={bf}}
\begin{tabular}{@{}lrrrr@{}}
\toprule
Case & $J_{u,\infty}$ & $J_{u,2}$ & $J_{q,\infty}$ & $J_{q,2}$ \\
\midrule
U & $3.109\times10^{-15}$ & $3.179\times10^{-16}$ & $5.329\times10^{-14}$ & $4.770\times10^{-15}$\\
C128 & {$8.666\times10^{-3}$} & {$3.459\times10^{-3}$} & {$4.652\times10^{-2}$} & {$2.423\times10^{-2}$} \\
C64 & {$1.862\times10^{-2}$} & {$8.424\times10^{-3}$} & {$4.176\times10^{-2}$} & {$1.654\times10^{-2}$} \\
C256 & {$1.068\times10^{-2}$} & {$4.391\times10^{-3}$} & {$4.589\times10^{-2}$} & {$1.652\times10^{-2}$} \\
CW40 & {$7.545\times10^{-3}$} & {$3.488\times10^{-3}$} & {$6.231\times10^{-2}$} & {$2.873\times10^{-2}$} \\
\bottomrule
\end{tabular}
\caption{Interface solution- and flux-jump errors for the UINN and standard
PINN cases on the circular-interface problem with
$(\beta^-,\beta^+)=(1,100)$. The RMS values $J_{u,2}$ and $J_{q,2}$
are arc-length weighted.}
\label{tab:circle-s-jumps}
\end{table}
\par
Table~\ref{tab:circle-s-jumps} shows that, unlike the UINN, all standard
PINN cases retain nonzero solution and flux jumps on the independent
interface evaluation set. Changing the number of interface collocation
points from $64$ to $256$ does not produce a systematic reduction in
these jump errors, while the wider CW40 case also retains noticeable 
jump errors. Thus, neither denser interface sampling nor increased network capacity enables the standard PINN to recover the exact jump satisfaction achieved by construction in the UINN.
\par
\begin{table}[!htbp]
\centering
\small
\renewcommand{\arraystretch}{1.13}
\footnotesize
\setlength{\tabcolsep}{3.5pt}
\ignorespaces
\captionsetup{font={small},labelfont={bf}}
\captionsetup{font={small},labelfont={bf}}
\begin{tabular}{@{}lrrrrrr@{}}
\toprule
Case & $\Omega^-\;(\%)$ & $\Omega^+\;(\%)$ & $B^-\;(\%)$ & $B^+\;(\%)$
& Boundary RMS & Boundary max \\
\midrule
U
& 0.58924 & 0.37152 & 0.43399 & 0.53905
& $1.269\times10^{-3}$ & $2.943\times10^{-3}$ \\
C128
& {3.82862} & {2.58518}
& {2.02929} & {3.18455}
& {$6.827\times10^{-3}$}
& {$2.079\times10^{-2}$} \\
C64
& {3.81338} & {2.38268}
& {1.80048} & {3.49683}
& {$7.406\times10^{-3}$}
& {$1.753\times10^{-2}$} \\
C256
& {3.88669} & {2.05998}
& {1.97728} & {3.24819}
& {$5.841\times10^{-3}$}
& {$1.452\times10^{-2}$} \\
CW40
& {1.77983} & {1.51465}
& {1.42986} & {2.23833}
& {$3.521\times10^{-3}$}
& {$6.458\times10^{-3}$} \\
\bottomrule
\end{tabular}
\caption{PDE residual and boundary errors corresponding to the UINN and
standard PINN solutions to the circular-interface problem with
$(\beta^-,\beta^+)=(1,100)$. The four PDE-residual columns report
relative RMS errors in percent, while the last two columns report
absolute boundary errors.}
\label{tab:circle-s-residuals}
\end{table}
\par

\par\smallskip{Table~\ref{tab:circle-s-residuals} shows that the larger
solution errors of the standard PINN cases are accompanied by larger
PDE residuals and boundary errors than those of the UINN.\par}
\par
\begin{table}[!htbp]
\centering
\small
\renewcommand{\arraystretch}{1.13}
\ignorespaces
\captionsetup{font={small},labelfont={bf}}
\captionsetup{font={small},labelfont={bf}}
\ignorespaces
\begin{tabular}{@{}lrrrr@{}}
\toprule
Case
& $E_{\Gamma,u,2}$
& $E_{\Gamma,u,\infty}$
& $E_{\Gamma,q,2}$
& $E_{\Gamma,q,\mathrm{rel}}$ (\%) \\
\midrule
U & $1.630\times10^{-4}$ & $3.471\times10^{-4}$
     & $1.192\times10^{-2}$ & 0.11927 \\
C128 & $4.619\times10^{-3}$ & $1.162\times10^{-2}$ & $1.701\times10^{-1}$ & 1.70165 \\
C64 & $8.644\times10^{-3}$ & $1.827\times10^{-2}$ & $1.919\times10^{-1}$ & 1.91966 \\
C256 & $5.010\times10^{-3}$ & $1.332\times10^{-2}$ & $1.791\times10^{-1}$ & 1.79162 \\
CW40 & $4.558\times10^{-3}$ & $8.877\times10^{-3}$ & $4.672\times10^{-2}$ & 0.46739 \\
\bottomrule
\end{tabular}
\caption{Errors of UINN and the standard PINN predictions against the exact
interface trace and normal flux for the circular-interface problem with
$(\beta^-,\beta^+)=(1,100)$. The quantities are defined in
\eqref{eq:interface-errors} using the minus-side prediction.}
\label{tab:circle-interface-accuracy}
\end{table}
\par

\par\smallskip{Table~\ref{tab:circle-interface-accuracy} addresses a different
question from the jump measures above: whether the predicted interface
trace and normal flux are themselves close to the exact physical data.
The UINN gives substantially smaller trace and flux errors than the
standard PINN cases, including both the approximately equal-parameter
C128 case and the wider CW40 case. Thus, in this test, the common
interface data produced by the projection--correction construction are
also more accurate, not merely mutually consistent across the two
phases.\par}

To visualize the behavior of the UINN and the primary standard PINN
case C128 (primary benchmark case) across the interface, Figure~\ref{fig:circle-cross} compares the solution and the corresponding
Cartesian flux component along the representative cross-section $y=0$. In panel (a), the UINN tracks the exact solution closely across the
interface, whereas C128 exhibits a clear solution error in $\Omega^-$,
with its profile lying noticeably below the exact solution. 
Because the C128 error in $\Omega^-$ is close to a constant offset,
it is largely removed by differentiation, so the $\beta u_x$ curves
in panel (b) nearly coincide.
Thus panel (a)
provides the clearer visualization of the accuracy improvement achieved
by the UINN, while panel (b) serves only as a complementary derivative
profile.

\begin{figure}[!htbp]\captionsetup{font={small},labelfont={bf}}\centering
\includegraphics[width=0.9\linewidth]{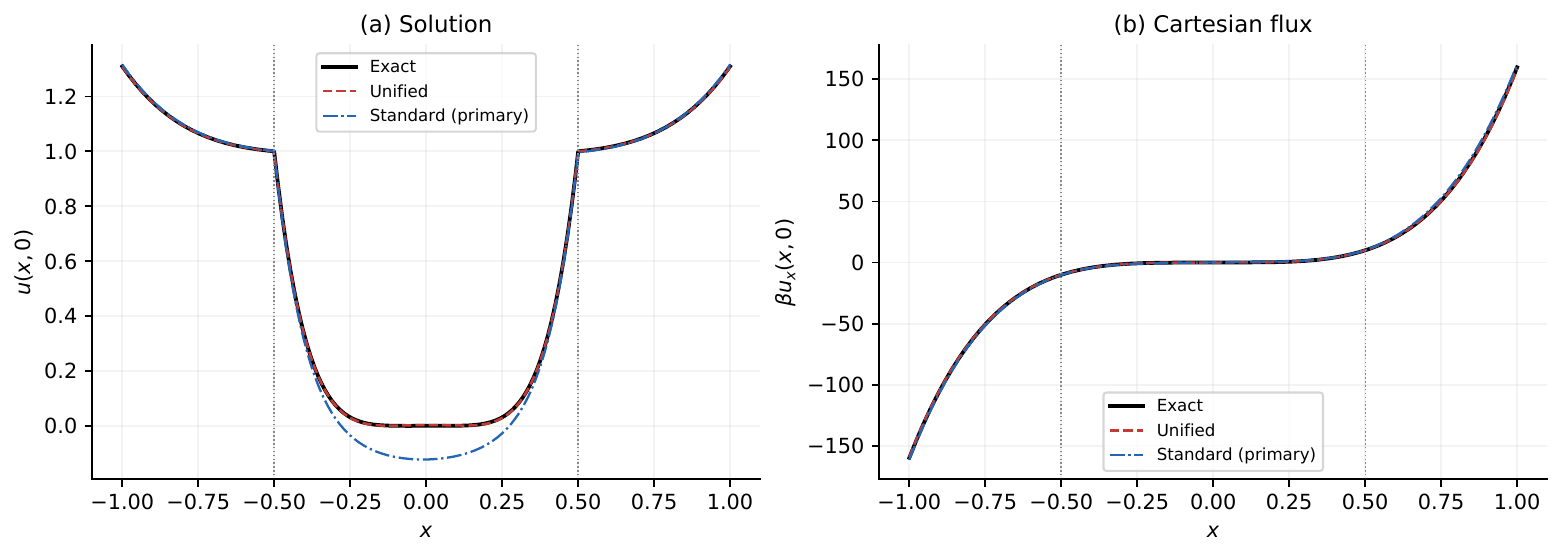}
\caption{Cross-sectional comparison for the circular-interface problem
with $(\beta^-,\beta^+)=(1,100)$ along $y=0$. Left: exact solution,
UINN, and the primary standard PINN case C128. Right: the corresponding
Cartesian flux component $\beta u_x$. Vertical dotted lines mark the
two intersections with $\Gamma$.}\label{fig:circle-cross}
\end{figure}

For the circular-interface experiment,
Tables~\ref{tab:circle-s-time} and~\ref{tab:circle-s-cost} quantify the
computational-cost measures. From
Table~\ref{tab:common-setup}, $N_I=1280$ and $N_b=256$. In the saved
UINN collocation set, $N_A=584$ PDE collocation points lie in the active
correction band, giving
\(
N_I+N_b+2N_A
=
1280+256+2(584)
=
2704
\)
network-point evaluations per loss evaluation, as reported in
Table~\ref{tab:circle-s-cost}. Table~\ref{tab:circle-s-time} reports the
corresponding optimization work and measured training time, while
Table~\ref{tab:circle-s-cost} reports the computational work associated
with one loss evaluation.
Tables~\ref{tab:circle-s-time} and~\ref{tab:circle-s-cost}
show that the UINN incurs higher computational cost and longer
training time than the standard PINN cases, primarily due to
projected interface evaluations and higher-order automatic
differentiation of the corrected field.
In the circular experiment, this additional cost is accompanied
by smaller solution errors, more accurate predictions of the
interface solution trace and physical normal flux
(Table~\ref{tab:circle-interface-accuracy}), and exact satisfaction
of both jump conditions by construction.
Possible strategies for reducing this overhead are discussed
in Section~\ref{sec:conclusion}.

\par
\begin{table}[!htbp]
\centering
\small
\renewcommand{\arraystretch}{1.13}
\footnotesize
\setlength{\tabcolsep}{3.5pt}
\ignorespaces
\captionsetup{font={small},labelfont={bf}}
\captionsetup{font={small},labelfont={bf}}
\begin{tabular}{@{}lrrrrrr@{}}
\toprule
Case & L-BFGS & Closures & Adam (s) & L-BFGS (s) & Total (s) & Final loss \\
\midrule
U & 1600 & 1828 & 208.78 & 200.32 & 409.10 & $3.030\times10^{-5}$ \\
C128 & {1600} & {1790} & {29.17} & {31.47} & {60.64} & {$1.378\times10^{-3}$} \\
C64 & {1600} & {1790} & {27.26} & {30.25} & {57.50} & {$1.327\times10^{-3}$} \\
C256 & {1600} & {1781} & {30.62} & {31.65} & {62.27} & {$1.240\times10^{-3}$} \\
CW40 & {1600} & {1741} & {39.36} & {41.29} & {80.65} & {$4.136\times10^{-4}$} \\
\bottomrule
\end{tabular}
\caption{Training time and optimization work for the UINN and standard
PINN cases on the circular-interface problem with
$(\beta^-,\beta^+)=(1,100)$. All cases use $2{,}000$ Adam updates and
$1{,}600$ L-BFGS iterations on the same one-thread AMD EPYC
configuration, with runs performed in separate sessions. Closures means the number of loss-and-gradient evaluations during L-BFGS. The reported
final losses correspond to the respective method-specific objectives.}\label{tab:circle-s-time}
\end{table}
\par
\par
\begin{table}[!htbp]
\centering
\small
\renewcommand{\arraystretch}{1.13}
\footnotesize
\setlength{\tabcolsep}{4pt}
\ignorespaces
\captionsetup{font={small},labelfont={bf}}
\captionsetup{font={small},labelfont={bf}}
\begin{tabular}{@{}lrrrr@{}}
\toprule
Case & Network points & Dense MACs ($10^6$) & Derivative order & Loss/backward (ms) \\
\midrule
U & 2,704 & 9.085 & 3 & 98.45 \\
C128 & 1,792 & 2.960 & 2 & {15.02} \\
C64 & 1,664 & 2.749 & 2 & {12.96} \\
C256 & 2,048 & 3.383 & 2 & {14.51} \\
CW40 & 1,792 & 5.949 & 2 & {18.91} \\
\bottomrule
\end{tabular}
\caption{Computational work and loss-and-backward timing for the
UINN and standard PINN cases on the circular-interface problem
with $(\beta^-,\beta^+)=(1,100)$.
Network-point counts and dense MACs refer to one loss evaluation,
with dense MACs counting linear-layer forward operations only.
The last column reports the time for evaluating the complete
training loss and backpropagating parameter gradients.}\label{tab:circle-s-cost}
\end{table}
\par

\FloatBarrier

\endgroup

\subsection{Star-shaped interface}
\label{sec:star}

Write $R=(x^2+y^2)^{1/2}$ and
$\vartheta=\operatorname{atan2}(y,x)$. The six-petal interface is
\begin{equation}
 F(x,y)=R^4\bigl(1+0.8\sin(6\vartheta)\bigr)-0.1=0,
 \qquad
 \rho(\vartheta)=
 \left(\frac{0.1}{1+0.8\sin(6\vartheta)}\right)^{1/4},
 \label{eq:star-interface}
\end{equation}
with $\Omega^-=\{R<\rho(\vartheta)\}$. The angular representation is
used to describe the interface, while the manufactured solution is
defined through the polynomial
\begin{equation}
 H(x,y)=10(x^2+y^2)^3
       +8(6x^5y-20x^3y^3+6xy^5)-(x^2+y^2).
 \label{eq:star-polynomial}
\end{equation}
For $R>0$, $H=10R^2F$, and hence $H=0$ on $\Gamma$.

For general diffusion coefficients $(\beta^-,\beta^+)$, we define
\begin{equation}
 u^s=1+\frac{4H}{\beta^s},\qquad
 \nabla u^s=\frac{4\nabla H}{\beta^s},\qquad
 f=16-1440(x^2+y^2)^2,\qquad
 g=1+\frac{4H}{\beta^+}\quad\text{on }\partial\Omega.
 \label{eq:star-exact}
\end{equation}
Since $\Delta H=360R^4-4$, these fields satisfy
$-\nabla\cdot(\beta\nabla u)=f$ in both phases. Moreover,
\[
u|_\Gamma=1,
\qquad
q_\Gamma=\beta^s\partial_{\bn}u^s
=4\nabla H\cdot\bn,
\]
so both the exact interface trace and the physical normal flux are
independent of the coefficient contrast.
In this section, we take
$(\beta^-,\beta^+)=(1,100)$.

\subsubsection{UINN experiments and results}
\label{sec:star-unified}

For the star-shaped interface, the UINN uses the same
projection--correction construction as in Section~\ref{sec:method}.
Since the interface is given explicitly in the radial form
$R=\rho(\vartheta)$, we take the signed radial defining coordinate
$\phi$ and the corresponding radial retraction $\cp$ as
\begin{equation}
 \phi(x,y)=R-\rho(\vartheta),\qquad
 \cp(x,y)=\rho(\vartheta)\boldsymbol e_r(\vartheta),
 \label{eq:star-geometry}
\end{equation}
where
$\boldsymbol{e}_r=(\cos\vartheta,\sin\vartheta)^T$ and
$\boldsymbol{e}_\vartheta=(-\sin\vartheta,\cos\vartheta)^T$.
With
$\ell=(\rho^2+(\rho')^2)^{1/2}$, the required geometric quantities are
\begin{equation}
 \bn=\frac{\rho\boldsymbol{e}_r-\rho'\boldsymbol{e}_\vartheta}{\ell},
 \qquad
 \kappa_\phi=\frac{\ell}{\rho},
 \qquad
 D\cp\,\bn=\bn-\kappa_\phi\boldsymbol{e}_r .
 \label{eq:star-geometry-factors}
\end{equation}
These quantities are substituted directly into
\eqref{eq:extension-coefficient}--\eqref{eq:corrected}.

\begin{figure}[!htbp]\centering
\includegraphics[width=\linewidth]{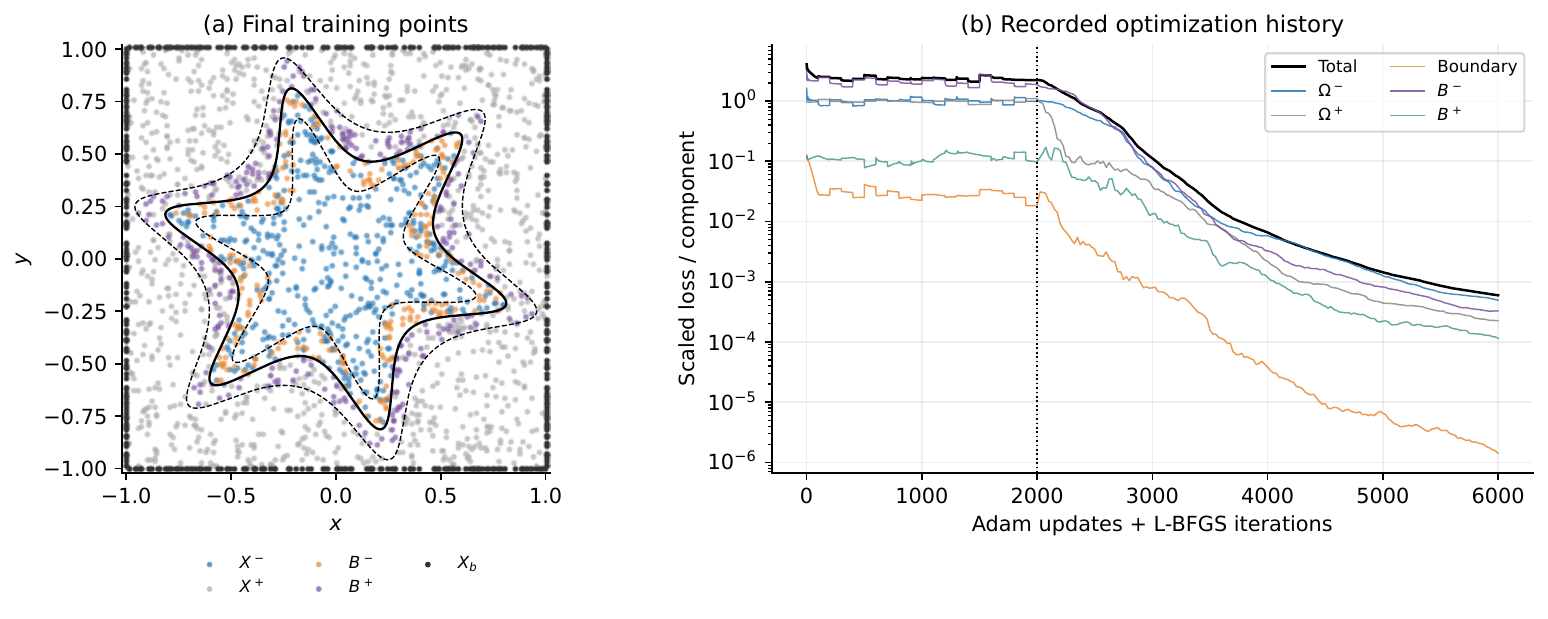}
\caption{UINN training for the star-shaped-interface problem with
$(\beta^-,\beta^+)=(1,100)$. Left: collocation points in
$\Omega^-$ and $\Omega^+$, in the two half-bands, and on the outer
boundary. The solid curve denotes $\Gamma$, and the dashed curves denote
the boundaries of the correction band $|\phi|<\delta$, with
$\delta=0.15$. Right: optimization history of the total loss and its
normalized residual components. The dotted vertical line separates Adam from L-BFGS.}\label{fig:star-u-training}
\end{figure}
The UINN is trained with $2,000$ Adam updates followed
by $4,000$ L-BFGS iterations. Figure~\ref{fig:star-u-training} shows the
collocation points and optimization history, while
Figure~\ref{fig:star-u-solution} shows the exact solution, UINN
approximation, and pointwise absolute error.
{Row ``U'' of Table~\ref{tab:star-s-errors} reports the resulting solution errors.} The phase-wise $L^2$ error is larger
in $\Omega^-$ than in $\Omega^+$, showing that the solution error is
concentrated primarily in the low-diffusion subdomain.

\begin{figure}[!htbp]\centering
\includegraphics[width=\linewidth]{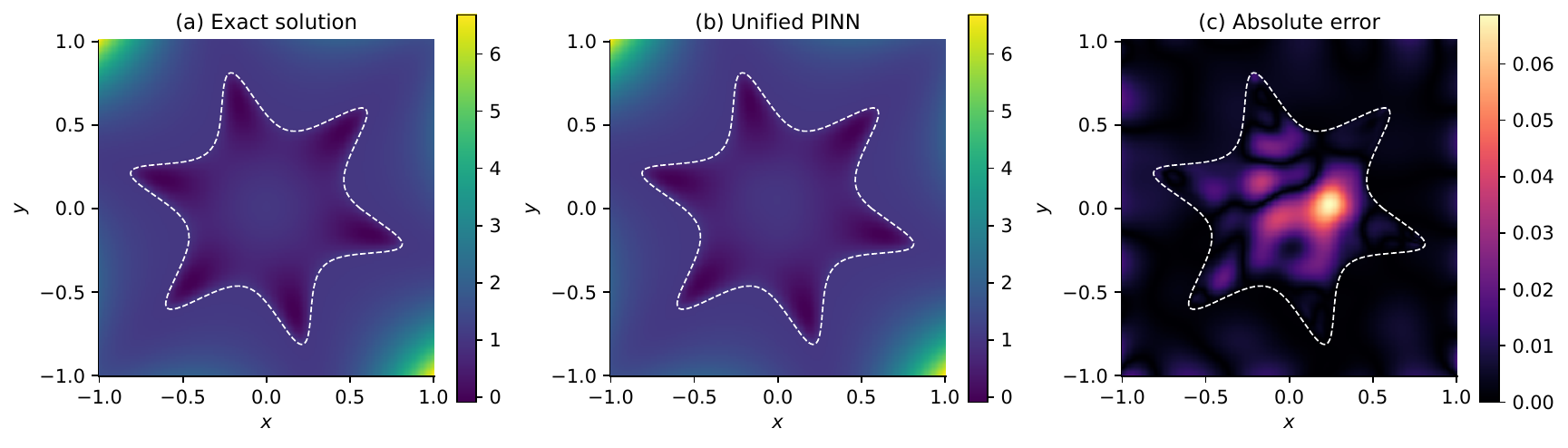}
\caption{Results on visualization grid for the star-shaped-interface problem, $(\beta^-,\beta^+)=(1,100)$: exact solution (left); UINN solution (middle), and pointwise absolute error (right). The dashed curve marks $\Gamma$.}
\label{fig:star-u-solution}
\end{figure}

{Row ``U'' of Table~\ref{tab:star-s-residuals} and the star-shaped
$(1,100)$ columns of Table~\ref{tab:uinn-raw-corrected-jumps}
summarize the PDE residuals and interface jump errors, respectively.} Similar to the observation in Section \ref{sec:circle}, the corrected
UINN satisfies both interface conditions to nearly machine precision, whereas
the raw branch fields retain noticeable solution and flux jump errors.
%
\subsubsection{Standard PINN benchmark tests}
\label{sec:star-standard}

For the star-shaped-interface standard PINN benchmark, S512 (with 512 interface collocation points) is
designated as the primary standard PINN case because its two
width-$28$ branches contain $3,474$ trainable parameters in total,
closely matching the $3,481$ parameters of the UINN. S256 and S1024
use the same architecture and differ only in the number of interface
collocation points (256 and 1024 points respectively), providing a check on the effect of interface
sampling. SW40 keeps $512$ interface points but uses width $40$ in
each branch, providing an equal-hidden-width (as UINN) comparison with increased
total parameter count. SP100 uses the same architecture and interface
sampling as S512 but increases both jump-loss weights from $10$ to
$100$, providing a check on the effect of stronger jump penalization.
All five standard cases use the same phase-wise, half-band, and
boundary collocation sets as the UINN and the same $2,000$ Adam
updates and $4,000$ L-BFGS iterations.
Table~\ref{tab:star-s-cases} summarizes these configurations.
\par
\begin{table}[!htbp]
\centering
\small
\renewcommand{\arraystretch}{1.13}
\ignorespaces
\captionsetup{font={small},labelfont={bf}}
\begin{tabular}{@{}lrrrrr@{}}
\toprule
Case & Hidden width & $N_\Gamma$ & Parameters & $\lambda_u$ & $\lambda_q$ \\
\midrule
U & 40 / shared & 0 & 3,481 & -- & -- \\
S512 & 28 / each & 512 & 3,474 & 10 & 10 \\
S256 & 28 / each & 256 & 3,474 & 10 & 10 \\
S1024 & 28 / each & 1024 & 3,474 & 10 & 10 \\
SW40 & 40 / each & 512 & 6,882 & 10 & 10 \\
SP100 & 28 / each & 512 & 3,474 & 100 & 100 \\
\bottomrule
\end{tabular}
\caption{Network, interface sampling, and penalty configurations for the star-shaped-interface problem with
$(\beta^-,\beta^+)=(1,100)$. S512 is the primary approximately
equal-parameter case; S256 and S1024 vary the number of interface
collocation points. SW40 uses the same hidden width as the UINN; SP100 increases interface penalty weights. All networks have three hidden layers.}
\label{tab:star-s-cases}
\end{table}
\par

Figure~\ref{fig:star-s-training} shows the collocation distribution
and optimization history for the primary S512 case, while
Figure~\ref{fig:star-s-solution} shows the exact solution, S512
approximation, and pointwise absolute error.

\begin{figure}[!htbp]\captionsetup{font={small},labelfont={bf}}\centering
\includegraphics[width=\linewidth]{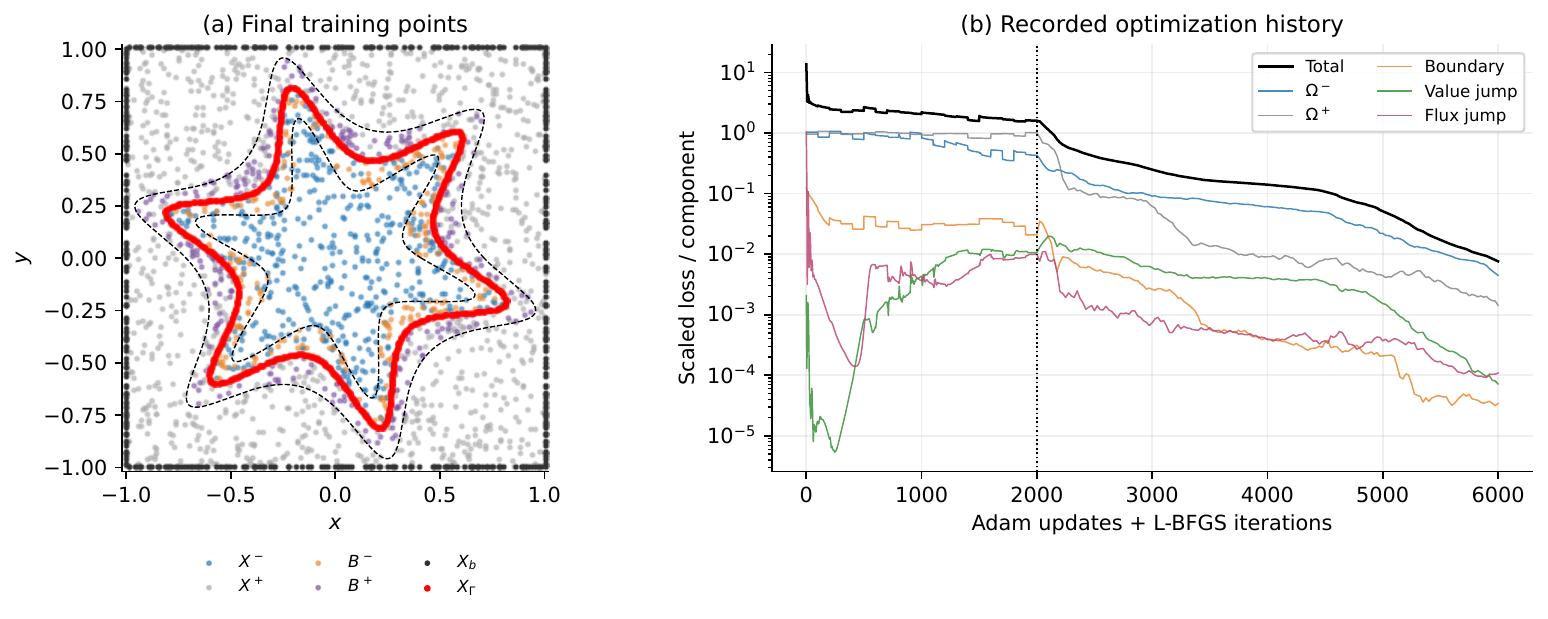}
\caption{Standard PINN case S512 for the star-shaped-interface problem
with $(\beta^-,\beta^+)=(1,100)$. Left: phase-wise, half-band,
boundary, and interface collocation points. Right: training history of
the total loss and its PDE, boundary, value-jump (solution-jump), and flux-jump components.}\label{fig:star-s-training}
\end{figure}

\par
\begin{table}[!htbp]
\centering
\small
\renewcommand{\arraystretch}{1.13}
\footnotesize
\setlength{\tabcolsep}{3.5pt}
\ignorespaces
\captionsetup{font={small},labelfont={bf}}
\begin{tabular}{@{}lrrrrrr@{}}
\toprule
Case & $E_0$ (\%) & $E_0^-$ (\%) & $E_0^+$ (\%) & $E_1$ (\%) & $E_E$ (\%) & $e_{h,\infty}$ \\
\midrule
U & 0.77382 & 2.75720 & 0.26945 & 2.18224 & 1.15946 & $6.864\times10^{-2}$ \\
S512 & 3.54195 & 12.79311 & 1.09148 & 9.61975 & 4.24068 & $1.942\times10^{-1}$ \\
S256 & 3.17113 & 11.91096 & 0.39023 & 4.67475 & 1.51013 & $2.106\times10^{-1}$ \\
S1024 & 3.63402 & 13.31295 & 0.93914 & 7.09168 & 3.70727 & $1.963\times10^{-1}$ \\
SW40 & 0.46373 & 1.67396 & 0.14376 & 1.04224 & 0.64360 & $2.247\times10^{-2}$ \\
SP100 & 33.54500 & 114.59420 & 14.93878 & 81.10927 & 40.87438 & $1.357\times10^{0}$ \\
\bottomrule
\end{tabular}
\caption{Solution error comparison between the UINN and standard PINN cases for
the star-shaped-interface problem with $(\beta^-,\beta^+)=(1,100)$.}
\label{tab:star-s-errors}
\end{table}
\par
\begin{figure}[!htbp]\captionsetup{font={small},labelfont={bf}}\centering
\includegraphics[width=\linewidth]{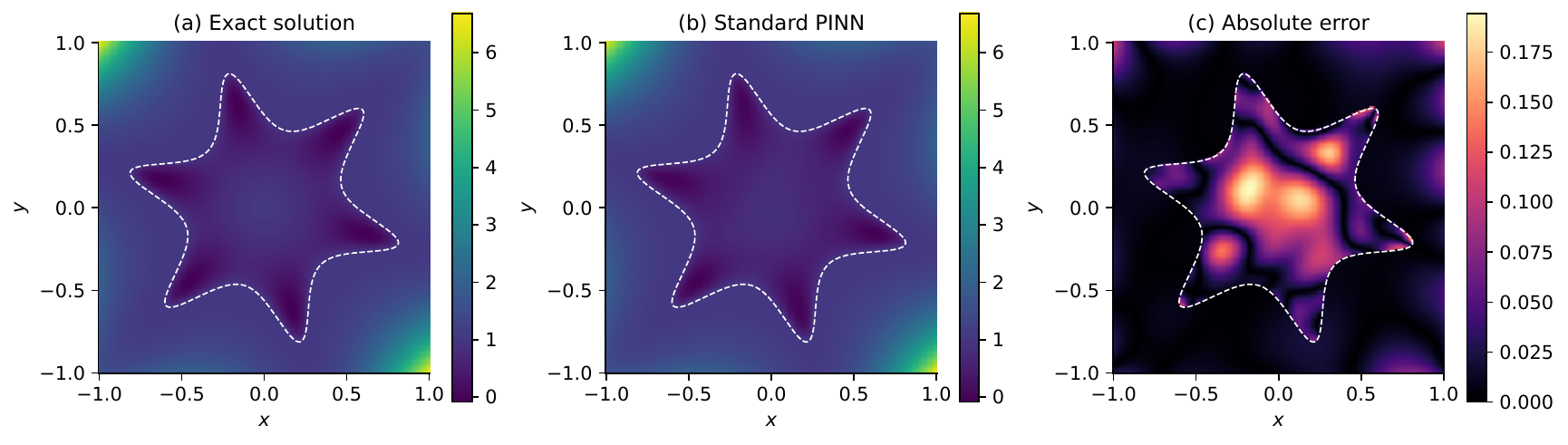}
\caption{Results for the standard PINN case S512 on the star-shaped-interface
problem with $(\beta^-,\beta^+)=(1,100)$: exact solution (left),
standard PINN solution (middle), and pointwise absolute error (right).}\label{fig:star-s-solution}
\end{figure}

Table~\ref{tab:star-s-errors} compares the solution errors of the UINN
with the five standard PINN cases. At approximately equal total
parameter count, the UINN is more accurate than the primary S512 case.
Changing only the number of interface collocation points does not
produce a systematic improvement among S256, S512, and S1024. The
wider SW40 case, which uses nearly twice as many trainable parameters,
achieves modestly smaller solution errors, while increasing the jump-loss weights
in SP100 substantially degrades the approximation.

Consistent with the circular-interface experiment, increasing interface sampling alone does not systematically improve standard PINN accuracy. A wider network can improve accuracy, whereas stronger jump penalties may degrade it. These observations highlight sensitivity to network capacity and penalty selection, while UINN avoids interface-penalty tuning altogether.\par
\begin{table}[!htbp]
\centering
\small
\renewcommand{\arraystretch}{1.13}
\ignorespaces
\captionsetup{font={small},labelfont={bf}}
\begin{tabular}{@{}lrrrr@{}}
\toprule
Case & $J_{u,\infty}$ & $J_{u,2}$ & $J_{q,\infty}$ & $J_{q,2}$ \\
\midrule
U & $2.487\times10^{-14}$ & $4.413\times10^{-15}$ & $2.043\times10^{-12}$ & $2.524\times10^{-13}$ \\
S512 & $1.890\times10^{-1}$ & $5.693\times10^{-2}$ & $1.581\times10^{0}$ & $3.430\times10^{-1}$ \\
S256 & $4.952\times10^{-2}$ & $1.746\times10^{-2}$ & $5.195\times10^{-1}$ & $1.314\times10^{-1}$ \\
S1024 & $1.671\times10^{-1}$ & $4.963\times10^{-2}$ & $9.339\times10^{-1}$ & $1.932\times10^{-1}$ \\
SW40 & $1.530\times10^{-2}$ & $5.059\times10^{-3}$ & $1.333\times10^{-1}$ & $3.884\times10^{-2}$ \\
SP100 & $9.768\times10^{-1}$ & $1.760\times10^{-1}$ & $1.366\times10^{0}$ & $4.097\times10^{-1}$ \\
\bottomrule
\end{tabular}
\caption{Interface solution- and flux-jump errors for the UINN and standard
PINN cases on the star-shaped-interface problem with
$(\beta^-,\beta^+)=(1,100)$. The RMS values $J_{u,2}$ and $J_{q,2}$
are arc-length weighted.}
\label{tab:star-s-jumps}
\end{table}
\par

Table~\ref{tab:star-s-jumps} further reveals a structural difference
between the two approaches. The UINN satisfies both jump conditions to
nearly machine precision, whereas all standard PINN cases retain finite jump
errors. In particular, although the wider SW40 case achieves modestly
smaller solution errors than the UINN, it does not capture interfacial jump conditions precisely. Thus, increasing approximation capacity can improve the
standard PINN solution without reproducing the exact interface jump
property built into the UINN.

Exact jump enforcement and solution accuracy are therefore viewed as distinct but related: satisfying the jump conditions exactly does not guarantee the smallest solution error. By eliminating jump penalties, UINN avoids their competition with PDE and boundary residuals, potentially improving optimization. In the circular test, UINN achieves smaller solution and interface errors than the standard PINNs, whereas the star-shaped test shows that a wider standard PINN can achieve moderately higher solution accuracy but still retains nonzero jumps and dependence on penalty weights.\par
\begin{table}[!htbp]
\centering
\small
\renewcommand{\arraystretch}{1.13}
\ignorespaces
\captionsetup{font={small},labelfont={bf}}
\begin{tabular}{@{}lrrrr@{}}
\toprule
Case & $E_{\Gamma,u,2}$
& $E_{\Gamma,u,\infty}$
& $E_{\Gamma,q,2}$
& $E_{\Gamma,q,\mathrm{rel}}~(\%)$ \\\midrule
U & $2.845\times10^{-3}$ & $7.244\times10^{-3}$ & $2.330\times10^{-1}$ & 1.22574 \\
S512 & $5.844\times10^{-2}$ & $1.905\times10^{-1}$ & $1.296\times10^{0}$ & 6.81611 \\
S256 & $1.831\times10^{-2}$ & $5.059\times10^{-2}$ & $4.214\times10^{-1}$ & 2.21692 \\
S1024 & $5.099\times10^{-2}$ & $1.777\times10^{-1}$ & $1.022\times10^{0}$ & 5.37759 \\
SW40 & $5.207\times10^{-3}$ & $1.568\times10^{-2}$ & $1.400\times10^{-1}$ & 0.73649 \\
SP100 & $2.787\times10^{-1}$ & $1.203\times10^{0}$ & $1.851\times10^{1}$ & 97.39772 \\
\bottomrule
\end{tabular}
\caption{Errors of UINN and the standard PINN predictions against the
exact interface trace and normal flux for the star-shaped-interface
problem with $(\beta^-,\beta^+)=(1,100)$. The quantities are defined
in \eqref{eq:interface-errors} using the minus-side prediction.}\label{tab:star-interface-accuracy}
\end{table}
\par

Table~\ref{tab:star-interface-accuracy} complements the jump-error
comparison by measuring accuracy against the exact interface data.
The UINN gives smaller solution-trace and normal-flux errors than the
approximately equal-parameter S512 case. The wider SW40 case gives a
slightly smaller normal-flux error than the UINN, consistent with its
smaller gradient-based errors in this test, while the UINN retains the
smaller solution-trace error and satisfies the jump conditions exactly.
These results indicate that exact jump enforcement and accuracy against
the exact interface data need not improve in precisely the same way.

\par
\begin{table}[!htbp]
\centering
\small
\renewcommand{\arraystretch}{1.13}
\footnotesize
\setlength{\tabcolsep}{3.5pt}
\ignorespaces
\captionsetup{font={small},labelfont={bf}}
\begin{tabular}{@{}lrrrrrr@{}}
\toprule
Case & $\Omega^-\;(\%)$ & $\Omega^+\;(\%)$ & $B^-\;(\%)$ & $B^+\;(\%)$
& Boundary RMS & Boundary max\\
\midrule
U     & 3.83317 & 1.96830 & 3.47715 & 3.64919
      & $7.881\times10^{-3}$ & $1.817\times10^{-2}$ \\
S512  & 7.67537 & 4.18761 & 5.51867 & 6.74182
      & $3.726\times10^{-2}$ & $1.179\times10^{-1}$ \\
S256  & 3.01563 & 2.26494 & 2.03323 & 1.90335
      & $1.019\times10^{-2}$ & $2.823\times10^{-2}$ \\
S1024 & 5.29605 & 3.07800 & 3.74411 & 5.57247
      & $2.848\times10^{-2}$ & $8.778\times10^{-2}$ \\
SW40  & 1.28775 & 1.10058 & 1.11284 & 1.29916
      & $4.436\times10^{-3}$ & $1.636\times10^{-2}$ \\
SP100 & 92.63533 & 36.54518 & 91.24783 & 57.98901
      & $4.087\times10^{-1}$ & $1.211\times10^{0}$ \\
\bottomrule
\end{tabular}
\caption{PDE residual and boundary errors corresponding to the UINN and
standard PINN solutions to the star-shaped-interface problem with
$(\beta^-,\beta^+)=(1,100)$. The four PDE-residual columns report
relative RMS errors in percent, while the last two columns report
absolute boundary errors.}
\label{tab:star-s-residuals}
\end{table}
\par
Table~\ref{tab:star-s-residuals} shows that the improved
solution accuracy of SW40 is accompanied by smaller PDE
residuals and boundary errors, consistent with its lower
normal-flux error in Table~\ref{tab:star-interface-accuracy}.
These results illustrate the benefits of increased network
capacity, but the improvement requires substantially more
trainable parameters and does not eliminate the nonzero
jump errors for SW40 reported in Table~\ref{tab:star-s-jumps}.
In contrast, UINN enforces both jump conditions exactly
by construction with fewer parameters and without
interface-penalty tuning.
Similar to the circular-interface experiment, the UINN incurs a higher
computational cost than the standard PINN because of the additional
projection--correction operations and higher-order automatic
differentiation. We therefore do not repeat the detailed timing and
computational-cost comparison here.

\FloatBarrier
\begingroup\ignorespaces
\subsection{Star-shaped interface for high-contrast coefficient}
\label{sec:star1000}

Here we consider the star-shaped-interface problem of
Section~\ref{sec:star} in a high-contrast setting with
$(\beta^-,\beta^+)=(1,1000)$, while keeping all other geometric and
problem settings unchanged. Throughout this subsection, the suffix
``-K'' in the figures and tables is used to distinguish the $(1{,}1000)$ cases from the corresponding
$(1{,}100)$ cases.

\subsubsection{UINN experiments and results}
\label{sec:star1000-unified}

As in Section~\ref{sec:star-unified}, the UINN uses
the same projection--correction construction,
geometric representation, network architecture,
and collocation sets. Training uses $2,000$ Adam updates
followed by $8,000$ L-BFGS iterations, while all other optimization
settings remain unchanged. Figure~\ref{fig:star1000-u-training} shows
the collocation distribution and optimization history, while
Figure~\ref{fig:star1000-u-solution} shows the exact solution, UINN
approximation, and pointwise absolute error.
{Row ``U-K'' of Table~\ref{tab:star1000-standard-errors} reports the
corresponding solution errors, where the relative $L^2$ error is only 0.7\%.}
\begin{figure}[!htbp]\captionsetup{font={small},labelfont={bf}}\centering
\includegraphics[width=\linewidth]{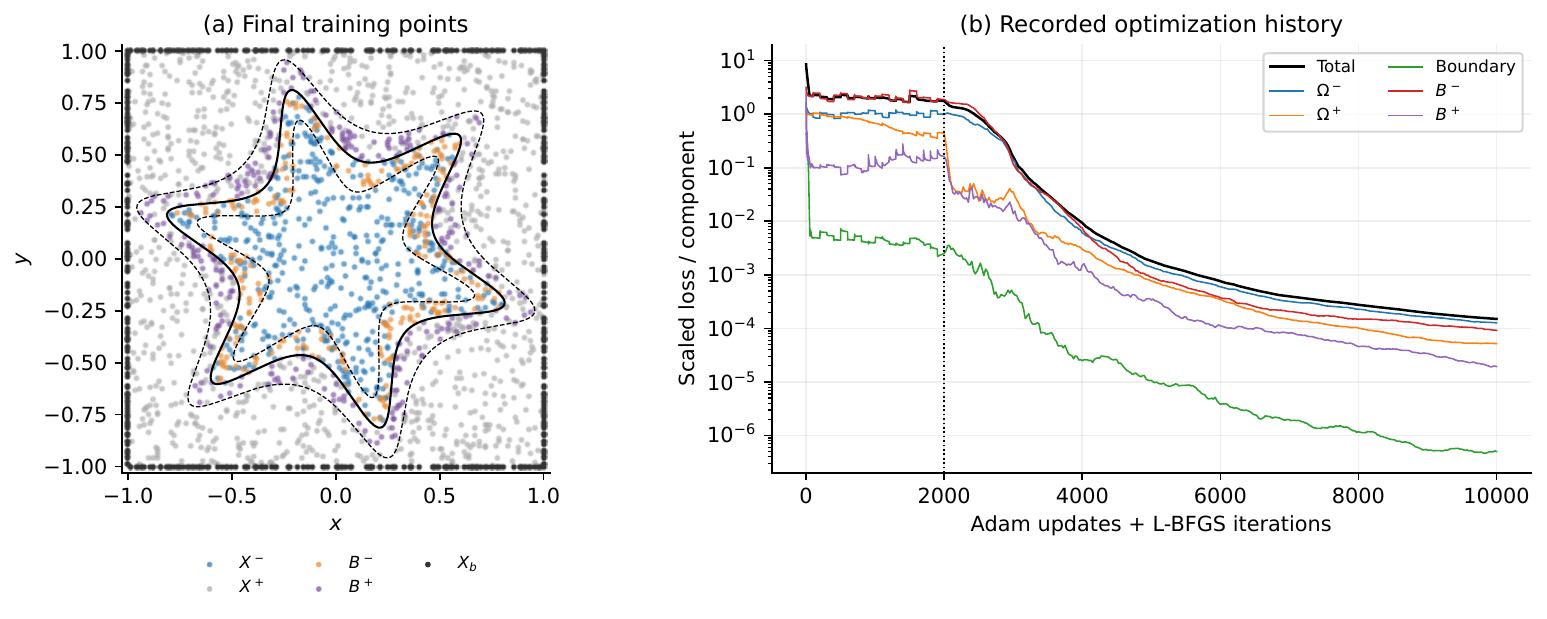}
\caption{UINN training for the star-shaped-interface problem with
$(\beta^-,\beta^+)=(1,1000)$. Left: collocation points in
$\Omega^-$ and $\Omega^+$, in the two half-bands, and on the outer
boundary. The solid curve denotes $\Gamma$, and the dashed curves denote
the boundaries of the correction band $|\phi|<\delta$, with
$\delta=0.15$. Right: optimization history of the total loss and its
normalized residual components. The dotted vertical line separates Adam from L-BFGS.}\label{fig:star1000-u-training}
\end{figure}

\begin{figure}[!htbp]\captionsetup{font={small},labelfont={bf}}\centering
\includegraphics[width=\linewidth]{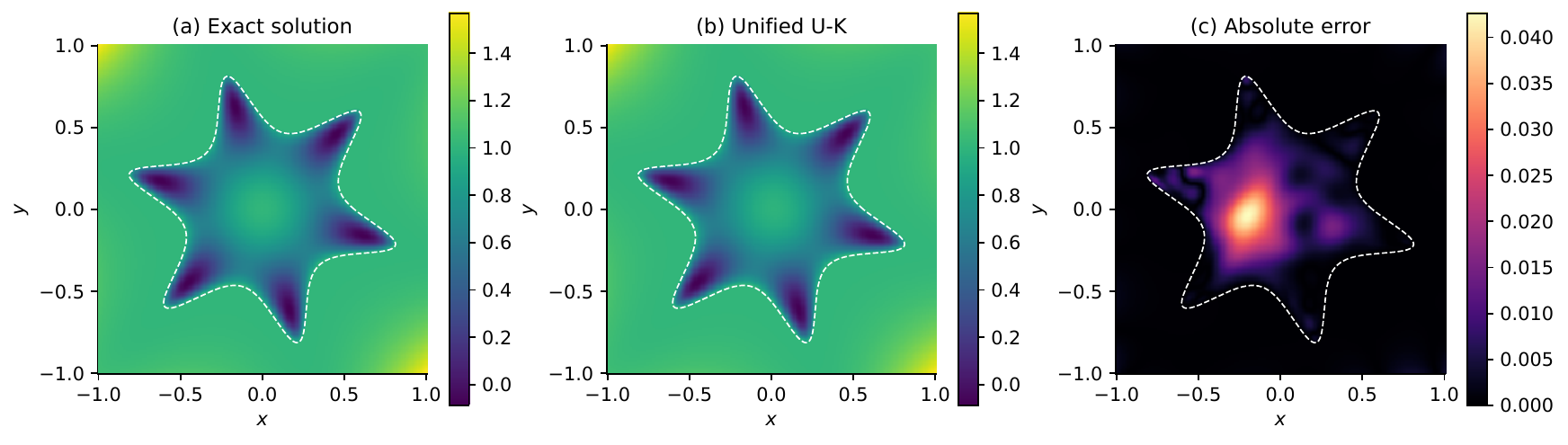}
\caption{Results on visualization grid for the star-shaped-interface problem, $(\beta^-,\beta^+)=(1,1000)$: exact solution (left); UINN solution (middle), and pointwise absolute error (right). The dashed curve marks $\Gamma$.}
\label{fig:star1000-u-solution}
\end{figure}

The UINN maintains good accuracy in the high-contrast case, with the
larger phase-wise $L^2$ error again concentrated in the low-diffusion
subdomain $\Omega^-$. Recall that the relative $L^2$ error $E_0$ and
its subdomain-restricted counterparts are defined in \eqref{eq:errors}.
For the exterior solution in \eqref{eq:star-exact},
$u^+=1+4H/\beta^+$, so when $\beta^+$ is large, $u^+$ is close to the
constant value $1$. Consequently, the denominator
$\|u\|_{L^2(\Omega^+)}$ in the standard relative error is dominated by
this constant part, and the resulting relative error does not directly
measure the accuracy relative to the much smaller variation $u^+-1$.
For this purpose, one may instead consider the shifted relative error
\[\widetilde E_0^+
=\frac{\|u_h-u\|_{L^2(\Omega^+)}}{\|u-1\|_{L^2(\Omega^+)}},\]
which is $0.56196\%$ in this test and still indicates good accuracy in
$\Omega^+$. The sampled maximum error $e_{h,\infty}$ provides an
absolute measure that does not depend on this choice of normalization.
Together, the relative and maximum errors show that the dominant error
remains concentrated in $\Omega^-$ while the UINN retains good overall
accuracy in the high-contrast case.

{Row ``U-K'' of Table~\ref{tab:star1000-standard-residuals} reports the independently}
evaluated PDE residuals for the high-contrast UINN, and the relative RMS
residuals remain below $3.3\%$ in both phases and their corresponding
half-bands. {The star-shaped $(1,1000)$ columns of Table~\ref{tab:uinn-raw-corrected-jumps} show} that the UINN
continues to enforce both jump conditions to nearly machine precision in the
high-contrast regime.
\endgroup

\subsubsection{Standard PINN benchmark tests}
\label{sec:star1000-standard}

We also use the same five standard PINN
configurations for the high-contrast test. S512-K is the primary
approximately equal-parameter case; S256-K and S1024-K vary the number
of interface collocation points, SW40-K increases the network width,
and SP100-K increases both jump-loss weights to $100$. All cases use
the same phase-wise, half-band, and boundary collocation sets as U-K (UINN case)
and are trained with $2,000$ Adam updates followed by $8,000$ L-BFGS
iterations. Table~\ref{tab:star1000-standard-cases} summarizes these
configurations.
\begin{table}[!htbp]
\centering\small
\ignorespaces
\captionsetup{font={small},labelfont={bf}}
\renewcommand{\arraystretch}{1.13}
\begin{tabular}{@{}lrrrrr@{}}
\toprule
Case & Hidden width & $N_\Gamma$ & Parameters & $\lambda_u$ & $\lambda_q$ \\
\midrule
U-K & 40 / shared & 0 & 3,481 & -- & -- \\
S512-K & 28 / each & 512 & 3,474 & 10 & 10 \\
S256-K & 28 / each & 256 & 3,474 & 10 & 10 \\
S1024-K & 28 / each & 1024 & 3,474 & 10 & 10 \\
SW40-K & 40 / each & 512 & 6,882 & 10 & 10 \\
SP100-K & 28 / each & 512 & 3,474 & 100 & 100 \\
\bottomrule
\end{tabular}
\caption{Network, interface sampling, and penalty configurations for the star-shaped-interface problem with
$(\beta^-,\beta^+)=(1,1000)$. S512-K is the primary approximately
equal-parameter case; S256-K and S1024-K vary the number of interface
collocation points. SW40-K uses the same hidden width as the UINN;
SP100-K increases interface penalty weights. All networks have three hidden layers.}
\label{tab:star1000-standard-cases}
\end{table}

\begin{figure}[!htbp]\captionsetup{font={small},labelfont={bf}}\centering
\includegraphics[width=\linewidth]{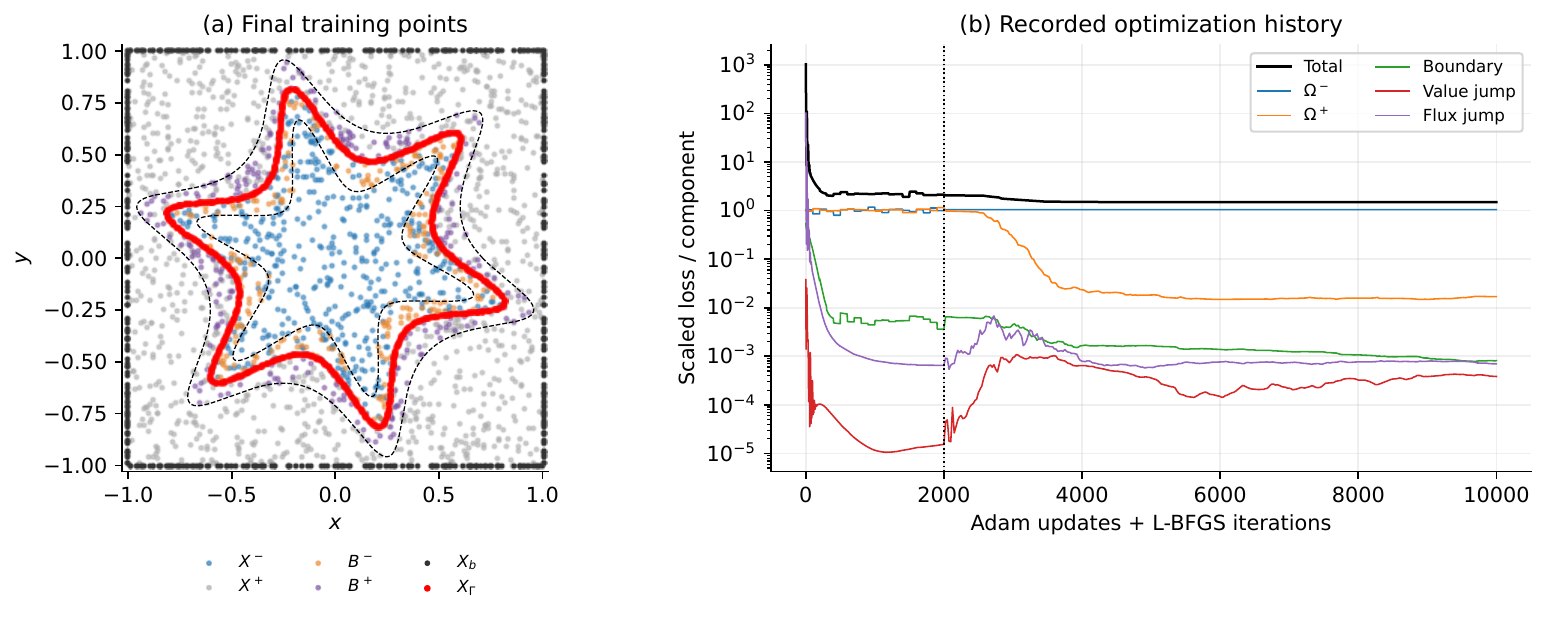}
\caption{Standard PINN case S512-K for the star-shaped-interface problem
with $(\beta^-,\beta^+)=(1,1000)$. Left: phase-wise, half-band,
boundary, and interface collocation points. Right: training history of
the total loss and its PDE, boundary, value-jump (solution-jump), and flux-jump components.}
\label{fig:star1000-s-training}
\end{figure}

\begin{table}[!htbp]
\centering\small
\ignorespaces
\captionsetup{font={small},labelfont={bf}}
\renewcommand{\arraystretch}{1.13}
\footnotesize\setlength{\tabcolsep}{3.4pt}
\begin{tabular}{@{}lrrrrrr@{}}
\toprule
Case & $E_0\;(\%)$ & $E_0^-\;(\%)$ & $E_0^+\;(\%)$ & $E_1\;(\%)$ & $E_E\;(\%)$ & $e_{\infty,h}$ \\
\midrule
{U-K} & {0.71298} & {1.95048} & {0.03805} & {1.36453} & {1.44939} & {$4.256\times10^{-2}$} \\
S512-K & 30.17216 & 82.42987 & 2.32323 & 99.07127 & 47.76575 & $1.117\times10^{0}$ \\
S256-K & 29.99153 & 81.98301 & 2.03906 & 98.96409 & 46.84629 & $1.116\times10^{0}$ \\
S1024-K & 30.20652 & 82.54468 & 2.20870 & 99.04428 & 46.31087 & $1.113\times10^{0}$ \\
SW40-K & 29.92244 & 81.82257 & 1.85024 & 98.99869 & 44.90179 & $1.113\times10^{0}$ \\
SP100-K & 32.06092 & 87.43113 & 3.22003 & 99.74500 & 57.73221 & $1.124\times10^{0}$ \\
\bottomrule
\end{tabular}
\caption{Solution error comparison between the UINN and standard PINN
cases for the star-shaped-interface problem with
$(\beta^-,\beta^+)=(1,1000)$.}
\label{tab:star1000-standard-errors}
\end{table}

Figure~\ref{fig:star1000-s-training} shows the collocation distribution
and optimization history for the primary S512-K benchmark case. Table~\ref{tab:star1000-standard-errors} shows a pronounced loss of
accuracy for all tested standard PINN configurations in the
high-contrast regime. Tables~\ref{tab:star1000-standard-jumps}--\ref{tab:star1000-standard-residuals}
further clarify the nature of this loss of accuracy.
Table~\ref{tab:star1000-standard-jumps} shows that stronger interface
penalties can substantially reduce the solution and flux jumps of the
standard PINN. However, Table~\ref{tab:star1000-interface-accuracy}
shows that this improved interfacial consistency does not translate
into an accurate physical interface flux. In particular, SP100-K
retains an interface-flux error of approximately $100\%$ despite its
smaller jump errors. Table~\ref{tab:star1000-standard-residuals}
further shows that the discrepancy is not confined to the interface:
the standard PINN cases also retain large post-training PDE residuals
in the subdomains. Thus, in this high-contrast test, reducing the
penalty-based interface mismatch alone is insufficient to recover an
accurate solution.

\begin{table}[!htbp]
\centering\small
\ignorespaces
\captionsetup{font={small},labelfont={bf}}
\renewcommand{\arraystretch}{1.13}
\begin{tabular}{@{}lrrrr@{}}
\toprule
Case & $J_{u,\infty}$ & $J_{u,2}$ & $J_{q,\infty}$ & $J_{q,2}$ \\
\midrule
{U-K} & {$3.020\times10^{-14}$} & {$5.196\times10^{-15}$} & {$3.524\times10^{-12}$} & {$3.909\times10^{-13}$} \\
S512-K & $6.181\times10^{-2}$ & $3.074\times10^{-2}$ & $3.103\times10^{0}$ & $8.669\times10^{-1}$ \\
S256-K & $7.466\times10^{-2}$ & $3.395\times10^{-2}$ & $3.990\times10^{0}$ & $9.378\times10^{-1}$ \\
S1024-K & $6.791\times10^{-2}$ & $3.059\times10^{-2}$ & $3.345\times10^{0}$ & $9.506\times10^{-1}$ \\
SW40-K & $1.147\times10^{-1}$ & $3.467\times10^{-2}$ & $2.647\times10^{0}$ & $6.818\times10^{-1}$ \\
SP100-K & $7.590\times10^{-3}$ & $2.757\times10^{-3}$ & $6.501\times10^{-1}$ & $1.550\times10^{-1}$ \\
\bottomrule\end{tabular}
\caption{Interface solution- and flux-jump errors for the UINN and standard PINN
cases on the star-shaped-interface problem with
$(\beta^-,\beta^+)=(1,1000)$. The RMS values $J_{u,2}$ and $J_{q,2}$
are arc-length weighted.}
\label{tab:star1000-standard-jumps}
\end{table}

\begin{table}[!htbp]
\centering\small
\ignorespaces
\captionsetup{font={small},labelfont={bf}}
\renewcommand{\arraystretch}{1.13}
\begin{tabular}{@{}lrrrr@{}}
\toprule
Case
& $E_{\Gamma,u,2}$
& $E_{\Gamma,u,\infty}$
& $E_{\Gamma,q,2}$
& $E_{\Gamma,q,\mathrm{rel}}\;(\%)$ \\
\midrule
{U-K}
& {$1.248\times10^{-4}$}
& {$2.668\times10^{-4}$}
& {$2.388\times10^{-1}$}
& {1.25618} \\
S512-K
& $3.926\times10^{-2}$
& $8.758\times10^{-2}$
& $1.887\times10^{1}$
& 99.26183 \\
S256-K
& $3.860\times10^{-2}$
& $9.182\times10^{-2}$
& $1.886\times10^{1}$
& 99.20628 \\
S1024-K
& $3.883\times10^{-2}$
& $9.107\times10^{-2}$
& $1.887\times10^{1}$
& 99.26012 \\
SW40-K
& $3.920\times10^{-2}$
& $1.310\times10^{-1}$
& $1.886\times10^{1}$
& 99.21700 \\
SP100-K
& $3.207\times10^{-2}$
& $3.607\times10^{-2}$
& $1.901\times10^{1}$
& 100.00458 \\
\bottomrule
\end{tabular}
\caption{Errors of UINN and the standard PINN predictions against the
exact interface trace and normal flux for the star-shaped-interface
problem with $(\beta^-,\beta^+)=(1,1000)$. The quantities are defined
in \eqref{eq:interface-errors} using the minus-side prediction.}
\label{tab:star1000-interface-accuracy}
\end{table}

\begin{table}[!htbp]
\centering\small
\ignorespaces
\captionsetup{font={small},labelfont={bf}}
\renewcommand{\arraystretch}{1.13}
\footnotesize\setlength{\tabcolsep}{3.4pt}
\begin{tabular}{@{}lrrrrrr@{}}
\toprule
Case & $\Omega^-\;(\%)$ & $\Omega^+\;(\%)$ & $B^-\;(\%)$ & $B^+\;(\%)$ & Boundary RMS & Boundary max \\
\midrule
{U-K} & {3.29284} & {0.78372} & {3.01668} & {1.99651} & {$1.079\times10^{-3}$} & {$3.634\times10^{-3}$} \\
S512-K & 99.73878 & 14.21052 & 99.76709 & 44.06995 & $4.198\times10^{-2}$ & $1.259\times10^{-1}$ \\
S256-K & 99.61573 & 12.83633 & 99.65123 & 43.46330 & $3.974\times10^{-2}$ & $1.345\times10^{-1}$ \\
S1024-K & 99.70112 & 13.76856 & 99.73077 & 42.90537 & $3.931\times10^{-2}$ & $1.183\times10^{-1}$ \\
SW40-K & 99.61554 & 12.82871 & 99.66845 & 41.62847 & $3.588\times10^{-2}$ & $1.037\times10^{-1}$ \\
SP100-K & 99.99777 & 18.50347 & 99.99799 & 46.81829 & $6.279\times10^{-2}$ & $1.919\times10^{-1}$ \\
\bottomrule\end{tabular}
\caption{PDE residual and boundary errors for the UINN and standard PINN
cases on the star-shaped-interface problem with
$(\beta^-,\beta^+)=(1,1000)$. The four PDE-residual columns report
relative RMS errors in percent, while the last two columns report
absolute boundary errors.}
\label{tab:star1000-standard-residuals}
\end{table}
\FloatBarrier
The high-contrast test reveals substantially different results for UINN and the standard PINNs under the common training protocol considered here. All tested standard PINN configurations exhibit large solution and
PDE residual errors, indicating that the difficulty is not limited to
imperfect enforcement of the jump conditions. This observation, however, \emph{should
not} be interpreted as implying that standard PINNs cannot approximate
the high-contrast problem. It merely indicates that the tested penalty-based
configurations do not recover ideal accuracy under the present
architectures, sampling strategy, penalty choices, and optimization
budgets. 

In contrast, the UINN retains good solution accuracy while
satisfying both jump conditions exactly by construction.
These results demonstrate the robust performance of UINN even in the high-contrast regime and highlight a key structural advantage: interface matching is removed from the optimization objective, avoiding competition between PDE and boundary residuals and separately weighted jump penalties. In contrast, the standard PINN must minimize these residuals simultaneously, and its performance under the tested configurations is sensitive to network capacity and penalty selection.

Similar to the experiments in Sections \ref{sec:circle}-\ref{sec:star}, the UINN incurs a higher
computational cost and longer training time than the standard PINN.
The detailed computational-cost comparison is therefore not repeated
here.

\section{Conclusion and Discussion}
\label{sec:conclusion}

This work develops UINN, a unified neural approximation with a single shared network in which both solution and flux jump conditions are enforced exactly by a projection--correction construction. The resulting formulation removes interface matching from the loss and shows robust and accurate performance across different interface geometries and coefficient contrasts, including the high-contrast regime.

The current implementation is more expensive than the standard PINN benchmarks because of projected interface evaluations and higher-order automatic differentiation. This overhead is not necessarily intrinsic to the projection-correction principle. In particular, analytic treatment of correction derivatives can reduce nested differentiation, while adaptive residual sampling has been shown to improve PINN efficiency by concentrating collocation points where they are most needed~\cite{LuDeepXDE2021,WuAdaptiveSampling2023}. More efficient automatic-differentiation and compiled/vectorized execution provide another practical route to reducing this overhead~\cite{BaydinAD2018}. 

This work lays the foundation for UINN as a general neural approximation
framework for interface problems, in which the network trial-space
functions are designed to capture the prescribed interface conditions.
The elliptic interface problem studied here serves as a methodological
cornerstone for this framework, realizing this network-architectural
principle through a unified branch representation and a
projection-correction mechanism. Future work will focus on
computationally optimized realizations, approximation and stability
analysis of the UINN trial space and projection-correction construction,
as well as extensions to moving-interface problems and more complex nonlinear and coupled systems.

\section*{Data availability statements}
The code and data that support the findings of this study are available on request from the corresponding author.
\section*{Declaration of generative AI and AI-assisted technologies in the \\manuscript preparation process}
During the preparation of this manuscript, the authors used OpenAI's ChatGPT and Codex for language polishing, proofreading, improving clarity and readability, and limited assistance with code debugging. Following the use of these tools, the authors reviewed and edited the relevant content as needed. All scientific ideas, methodologies, algorithms, experimental designs and implementations, benchmark evaluations, and conclusions were independently developed and validated by the authors. The use of AI tools did not determine the scientific content or interpretation of the results. The authors take full responsibility for the content of the manuscript.

\bibliographystyle{plain}

\end{document}